\def\isarxivversion{1} 

\ifdefined\isarxivversion

\documentclass[11pt]{article}

\else

\documentclass[twoside]{article}
\usepackage[accepted]{aistats2020}

\fi

\usepackage{amsmath}
\usepackage{amsthm}
\usepackage{amssymb}
\usepackage{algorithm}
\usepackage{subfig}
\usepackage{color}
\usepackage[english]{babel}
\usepackage{graphicx}
\usepackage{wrapfig,epsfig}
\usepackage{epstopdf}
\usepackage{url}
\usepackage{graphicx}
\usepackage{color}
\usepackage{epstopdf}
\usepackage{algpseudocode}
\usepackage{scrextend}
\usepackage[T1]{fontenc}
\usepackage{bbm}
\usepackage{comment}

\usepackage{tikz}
\usepackage{hyperref}  
\hypersetup{colorlinks=true,citecolor=blue,linkcolor=blue} 
\usetikzlibrary{arrows}
\usepackage[margin=1in]{geometry}
\graphicspath{{./figs/}}

\newtheorem{theorem}{Theorem}[section]
\newtheorem{lemma}[theorem]{Lemma}
\newtheorem{definition}[theorem]{Definition}

\newtheorem{fact}[theorem]{Fact}
\newtheorem{remark}[theorem]{Remark}

\newtheorem{problem}[theorem]{Problem}

\newcommand{\wh}{\widehat}
\newcommand{\wt}{\widetilde}

\newcommand{\R}{\mathbb{R}}

\renewcommand{\varepsilon}{\epsilon}
\renewcommand{\tilde}{\wt}

\DeclareMathOperator*{\E}{{\bf {E}}}

\DeclareMathOperator*{\Z}{\mathbb{Z}}

\DeclareMathOperator{\OPT}{OPT}
\DeclareMathOperator{\supp}{supp}

\DeclareMathOperator{\poly}{poly}
\DeclareMathOperator{\Tr}{tr}
\DeclareMathOperator{\nnz}{nnz}

\DeclareMathOperator{\rank}{rank}

\newcommand{\SQDATA}{\textsc{SqrtData}}
\newcommand{\LOGDATA}{\textsc{LogData}}

\makeatletter
\newcommand*{\RN}[1]{\expandafter\@slowromancap\romannumeral #1@}
\makeatother

\usepackage{lineno}

\usepackage{etoolbox}

\makeatletter

\newcommand{\define}[4][ignore]{%
  \ifstrequal{#1}{ignore}{}{
  \@namedef{thmtitle@#2}{#1}}%
  \@namedef{thm@#2}{#4}%
  \@namedef{thmtypen@#2}{lemma}%
  \newtheorem{thmtype@#2}[theorem]{#3}%
  \newtheorem*{thmtypealt@#2}{#3~\ref{#2}}%
}

\newcommand{\state}[1]{%
  \@namedef{curthm}{#1}
  \@ifundefined{thmtitle@#1}{
  \begin{thmtype@#1}
    }{
  \begin{thmtype@#1}[\@nameuse{thmtitle@#1}]
  }
    \label{#1}
    \@nameuse{thm@#1}
  \end{thmtype@#1}
  \@ifundefined{thmdone@#1}{
  \@namedef{thmdone@#1}{stated}%
  }{}
}

\newcommand{\restate}[1]{%
  \@namedef{curthm}{#1}
  \@ifundefined{thmtitle@#1}{
    \begin{thmtypealt@#1}
    }{
  \begin{thmtypealt@#1}[\@nameuse{thmtitle@#1}]
  }
    \@nameuse{thm@#1}
  \end{thmtypealt@#1}
  \@ifundefined{thmdone@#1}{
  \@namedef{thmdone@#1}{stated}%
  }{}
}

\newcommand{\thmlabel}[1]{
  \@ifundefined{thmdone@\@nameuse{curthm}}{\label{#1}
    }{\tag*{\eqref{#1}}}
}
\makeatother

\begin{document}

%

%

\ifdefined\isarxivversion

\title{Sketching Transformed Matrices with Applications to\\ Natural Language Processing\thanks{A preliminary version of this paper appeared in the proceedings of AISTATS 2020.}}

\author{
Yingyu Liang\thanks{\texttt{yliang@cs.wisc.edu}. University of Wisconsin-Madison.}
\and 
Zhao Song\thanks{\texttt{zhaos@ias.edu}. Princeton University and Institute for Advanced Study.}
\and
Mengdi Wang\thanks{\texttt{mengdiw@princeton.edu}. Princeton University.}
\and
Lin F. Yang\thanks{\texttt{linyang@ee.ucla.edu}. University of California, Los Angeles.}
\and
Xin Yang\thanks{\texttt{yx1992@cs.washington.edu}. University of Washington.}
}

\date{}

\else

\twocolumn[

\aistatstitle{Sketching Transformed Matrices with Applications to Natural Language Processing}

\aistatsauthor{ 
Yingyu Liang\thanks{haha}
\and 
Zhao Song
\and
Mengdi Wang
\and
Lin F. Yang
\and
Xin Yang\thanks{\texttt{yx1992@cs.washington.edu}}
}

\aistatsaddress{ Address} 
]

\fi

\ifdefined\isarxivversion
\begin{titlepage}
  \maketitle
  \begin{abstract}
Suppose we are given a large matrix $A=(a_{i,j})$ that cannot be stored in memory but is in a disk or is presented in a data stream. However, we need to compute a matrix decomposition of the entry-wisely transformed matrix, $f(A):=(f(a_{i,j}))$ for some function $f$. Is it possible to do it in a space efficient way? Many machine learning applications indeed need to deal with such large transformed matrices, for example word embedding method in NLP needs to work with the pointwise mutual information (PMI) matrix, while the entrywise transformation makes it difficult to apply known linear algebraic tools. Existing approaches for this problem either need to store the whole matrix and perform the entry-wise transformation afterwards, which is space consuming or infeasible, or need to redesign the learning method, which is application specific and requires substantial remodeling.

In this paper, we first propose a space-efficient sketching algorithm for computing the product of a given small matrix with the transformed matrix. It works for a general family of transformations with provable small error bounds and thus can be used as a primitive in downstream learning tasks. 
We then apply this primitive to a concrete application: low-rank approximation. 
We show that our approach obtains small error and is efficient in both space and time. 
We complement our theoretical results with experiments on synthetic and real data. 


  \end{abstract}
  \thispagestyle{empty}
\end{titlepage}

{\hypersetup{linkcolor=black}
\tableofcontents
}
\newpage

\else

\begin{abstract}

\end{abstract}

\fi

\section{Introduction}\label{sec:intro}
Matrix datasets are ubiquitous in machine learning. 
However, many matrix datasets are usually too large to fit in the computer memory in large scale applications, e.g., image clustering~\cite{pauca2006nonnegative}, natural language processing~\cite{michel2011quantitative}, network analysis~\cite{mao2004modeling,grover2016node2vec}, and recommendation systems~\cite{koren2009matrix}. 
Many techniques have been proposed to perform the learning tasks on these data in an efficient way; see, e.g.,~\cite{mahoney2011randomized,w14,zhou2008large,gemulla2011large} and the references therein. 
However, challenges arise when the learning task is performed on an entrywise transformation of the matrix, which prevents applying many linear algebraic techniques.
Furthermore, due to large sizes,
these matrices are often constructed by entrywise updates, i.e., the entries of the matrix are constructed from a stream of updates where each update adds some value on some entry.
More specifically, 
there is a very large underlying matrix $A$ (that cannot be stored in memory easily) whose entries are constructed by a data stream where each item in the stream is of the form $(i,j, \Delta)$ with $\Delta\in \{\pm 1\}$ representing the update $A_{i,j} \leftarrow A_{i,j} + \Delta$. The downstream learning task (e.g., low rank approximation), however, needs to take input as matrix $M$ where $M_{i,j} = f(A_{i,j})$ for some transformation function $f$ (e.g., $f(x)=\log(|x|+1)$).

A concrete example is word embedding in natural language processing (NLP). Word embedding methods aim to embed each word to a vector space. It becomes a basic building block in many modern NLP systems. 
Many of these systems achieve the state of the art performance on various tasks via word embedding~\cite{pennington2014glove,mikolov2013distributed,wu2016google}.
A basic routine in word embedding is to explicitly or implicitly perform low rank approximation of an entry-wise transformed matrix~\cite{lg14,lzm15}.
For instance, the transformation is to apply a log likelihood function on each entry. The matrix itself is the so-called co-occurrence count matrix, which can be constructed by scanning the text corpus, e.g., the entire Wikipedia database. This matrix is usually of size millions by millions. 

Similar examples include regressions on huge accumulated datasets in economics~\cite{dv13,v14b}, where different transformations on covariates are often used to reduce biases.
Other examples include visual feature extraction~\cite{boureau2010theoretical}, kernel methods~\cite{rahimi2008random}, and $M$-estimators~\cite{zhang1997parameter}.
These large scale applications make it  impractical or hard to implement existing methods, which keep the matrix in memory. 
Some other approaches exploit the problem structure to get around the huge space requirement.
For instance, some of them propose sequential models of the data, and design online algorithms for computing the embeddings (e.g.,\cite{mikolov2013distributed,bojanowski2016enriching}). 
These methods, however, are more task-specific and cannot be applied to other tasks involving more general entrywise matrix transformations.


In this paper,
we show that learning based on transformed large matrices is possible even when storing such a matrix is not feasible.
Our main contributions are:
\begin{itemize}
\item For a general class of transformation function $f$, we provide an efficient one-pass \emph{matrix-product sketch} for computing the product of a given small matrix $B$ with the transformed matrix $f(A)$ with provable error bounds.
This algorithm uses space at most the size of the output.
The method assumes no statistical model about the updates and can handle a general family of transformations.
In particular, these transformations include  logarithmic functions and small degree polynomials.
This method can also be used as building blocks for downstream tasks: any algorithm requires access to the transformed matrix via a matrix product can apply our algorithm to obtain space saving.
\item 
We demonstrate the application of our algorithm in a concrete task: low rank approximation.
To the best of our knowledge,
our algorithm is the first one that is able to compute low rank approximation of large matrices under entrywise transformations.
We plug in our matrix product sketch into known algorithms as black boxes. 
We provide theoretical analysis on the tradeoff between the space and the accuracy of these algorithms.
We show that our algorithms are space efficient and almost match the accuracy of using the full matrix. 
These theoretical guarantees are complemented by experiments for low rank approximation on synthetic and real data. 
The empirical results show that our algorithm can reduce the space usage by orders of magnitude while the error is almost the same as the optimum.
We show that our algorithms beat the baseline of using uniform sampling on columns of the transformed matrix by a large margin.
We also provide results on linear regression in the appendix.
\end{itemize}

\paragraph{Road Map.} We provide definitions and basic concepts in Section~\ref{sec:preli}.
In Section~\ref{sec:matrix_product}, we introduce our basic routine called the \emph{matrix product} sketch.
We use our sketching algorithms to compute the low rank approximation of a transformed matrix in Section~\ref{sec:lowrank},
and the application on linear regression is in Appendix \ref{sec:more application}.
In Section~\ref{sec:exp}, we use numeric experiments to justify our approach.
The appendix provides a list of related works, the complete proofs, details of the experiments, and also additional theoretical and empirical results. 
\section{Related Work} \label{sec:relatedwork}

There exists a large body of work on fast algorithms for large scale matrices. Some are based on randomized matrix algorithms and use techniques like sampling and sketching; see~\cite{mahoney2011randomized,w14} and the reference therein. Some others are based on optimization algorithms like Alternating Least Square and Stochastic Gradient Descent and their variants; see ~\cite{zhou2008large,gemulla2011large} for some examples. However, most existing approaches do not apply to the settings considered in this paper. The closest work is \cite{woodruff2016distributed}, which considers low rank approximation of the element-wise transformation of the sum of several matrices located in different machines. This distributed setting is different from our setting and na\"ively applying their algorithm will lead to a large space cost. Furthermore, our sketching method can be applied to learning tasks beyond low rank approximation. 

Our work is built on techniques from numerical linear algebra and streaming data analysis in the recent decade.
There are numerous research works along this line.
Here we list a few but far from exhaustive.

Low-rank approximation or  matrix factorization of a matrix  is an important task in numerical linear algebra.
In this problem, we are given a $n \times d$ matrix $A$ and a parameter $k$, the goal is to find a $\rank$-$k$ matrix $\wh{A}$ so as to minimize the residual error $\| \wh{A} - A \|_F^2$, where the Frobenius norm is defined as $\| A\|_F= ( \sum_{i=1}^n \sum_{j=1}^d A_{i,j}^2 )^{\frac{1}{2}}$. 
Note that an optimal $\wh{A}$ provides a good estimation to the leading eigenspace of the matrix $A$.
Classical way of speeding up low-rank approximation via sketching requires showing two properties for sketching matrix:  subspace embedding \cite{s06,lww19,ww19} and  approximate matrix product \cite{nn13,kn14}. 
Low-rank approximation algorithm via combining those two properties has been presented in several papers \cite{cw13,mm13,swz19}. The classical sketching idea is easy to be made a streaming algorithm, since we usually use linear sketching matrix, which we don't need to explicitly write down during the stream.
However none of these methods are applicable to our setting,
which is much harder than the classical streaming low-rank approximation problem. 
This is mainly because the transformation $f$ that acts on an the matrix $A$ completely destroyes the linear algebraic property of matrix $A$; see Appendix~\ref{sec:diff} for some discussions.
The storage of $A$ can also be indefeasibly large to be stored and apply the above mentioned methods.

Streaming algorithms have gained great progress since its first systematic study by \cite{ams99}. 
Classic streaming problems ask how to estimate a function over a vector, which is under streaming updates.
For instance, \cite{ams99} approximates $\|v\|_p$ while  observing a sequence of updates to the coordinates of $v$.
The usual assumption is that $v\in \R^n$ and $n$ is so large that $v$ cannot be stored in memory easily.
Since \cite{ams99}, a line of research works (e.g. \cite{i00, iw05, bks02, bksv14, knw10}) gradually improve the algorithm and obtain nearly optimal upper and lower bounds.
Very recently, \cite{bo10a,bo10b, bvwy17} attempts to handle a more general set of functions.
\cite{bvwy17} gives a nearly optimal characterization of this problem.
\cite{bbcky17} studies a more general setting, i.e., functions that do not have a summation structure $f:\R^n\rightarrow \R$.
They give optimal characterization for streaming all symmetric norms.
Given theses advances, none of them solves our problem directly since a  streaming estimation only gives a value of vector, that is unrelated to the matrix formulation of the input.

\section{Preliminaries}\label{sec:preli}

\textbf{Notation.}
$[n]$ denotes the set $\{1,2,\cdots,n\}$. 
For a vector $ x\in \mathbb{R}^n$, $|x| \in \R^n$ denotes a vector whose $i$-th entry is $|x_i|$.
For a matrix $A\in \R^{n\times n}$, let $\| A\|$ denote its spectral norm, $\sigma_i(A)$ to denote its $i$-th largest singular value, and $[A]_k$ denote its best rank-$k$ approximation. Also let $\det(A)$ denote its determinant when $A$ is square.  
For a function $f$, $M=f(A)$ means entrywise transformation $M_{ij} = f(A_{ij})$. 
We also denote $A_{i*}$ as the $i$-th row of matrix $A$ and $A_{*j}$ as its $j$-th column.


\textbf{Problem Definition.}
The problem of interests is defined as follows. Suppose we have a underlying large matrix $A=(A_{i,j})\in \R^{n\times n}$ initialized as a zero matrix.\footnote{Our method also applies to non-square $A$; we consider square matrices for simplicity.}
Now, we have observed a sequence of updates of the form 
$
\langle (i_1, j_1, 
\Delta_1), (i_2, j_2, \Delta_2), \ldots, (i_m, j_m, \Delta_m) \rangle
$
for some $m=\poly(n)$, $i_t, j_t\in [n]$ and $\Delta_t\in \{-1, 1\}$. 
At the $t$-th update, we are updating the underlying matrix by $a_{i_t, j_t}\gets a_{i_t,j_t}+ \Delta_t$.
We assume that $m$ is bounded by $\poly(n)$.
Note that the assumptions of integer updates is without loss of generality.
For instance, if the updates is not an integer, we can round them to a specified precision $\epsilon>0$ and then scale them to integers.
The polynomially bounded length is also a usual and reasonable assumption. 
At the end of the stream, one would like to perform some learning task (such as low-rank approximation) on the matrix $M=f(A)$ for some fixed function $f:\R\rightarrow\R$ and would like to do so using as small space as possible, in particular, avoid storing the large matrix $A$.
Some examples of the transformation functions are
\begin{align}
	f(x) = \log (|x|+1), \text{~or~}
	f(x) = |x|^{\alpha}, ~ \forall \alpha\ge 0.
\end{align}
Functions of this form are important in machine learning. For example, $f(x) =\log(|x|+1)$ corresponds to the log likelihood function and $f(x) = |x|^\alpha$ corresponds to a general family of statistic models or feature expansion.

In this paper we would like to design a space efficient method for approximating $Z=f(A) B$ for a given matrix $B$, where $f(A)\in \R^{n\times n}$ and $B\in \R^{n\times k}$ for some integer $n$ and $k$ with $k\ll n$. 
We would like to design algorithms that uses space $\wt{O}(nk)$ instead of $\wt{O}(n^2)$.
This can then be used as a plug-in primitive and turn learning algorithms into space efficient ones if they only access $f(A)$ by matrix product with small $B$. 
More formally,

\begin{problem}[approximate transformed matrix and matrix product]
Given a fixed matrix $B$ and function $f:\R\rightarrow \R$, design an algorithm that makes a single pass over an update stream of a matrix $A$, output an approximated value of $f(A) B$ with high probability. We require the algorithm to use as small space as possible (without counting the space of $B$).
\end{problem}

We call our method the \emph{sketch for $f$-matrix product}. 
We then demonstrate its effectiveness in the applications of linear regression and low rank approximation on $M=f(A)$. 
Linear regression is to minimize $\|Mx - b\|_2^2$, and low rank approximation is defined as follows.

\begin{problem}[low-rank approximation]
Given integers $k \le n$, an $n \times n$ matrix $M$, two parameters $\epsilon, \delta>0$, the goal is to output an orthonormal $n\times k$ matrix $L$ such that
\begin{align*}
\| L L^\top M - M \|_F^2 \leq (1+\epsilon) \| M - [M]_k \|_F^2 + \delta.
\end{align*} 
where $[M]_k = \arg\min_{\rank-k~M'} \| M - M' \|_F^2$.
\end{problem}



\section{Sketch for $f$-Matrix Product }

\label{sec:matrix_product}
Our goal in this section is to compute the matrix product $f(A)B$ where $B$ is given and $A$ is under updating or can only be read entry by entry.
We observe that each entry of $Z=f(A)B$ can be written as a vector product: $Z_{i,j} = \langle f(A)_{i*}, B_{*j} \rangle$.
Thus, we will first design a primitive to compute each $Z_{i,j}$ using small space.
Running a primitive in parallel for each entry $Z_{i,j}$ results in our full algorithm for computing the matrix product.
In the following sections, we will first introduce the vector sketch problem and present our vector product primitives for different functions $f$.
Lastly, we will combine them to form a unified algorithm for matrix product.

\subsection{Sketch for $f$-Vector Product}
Recall that for given vectors $x,y\in \R^n$, the inner product is defined as
$\langle x, y\rangle=\sum_{i=1}^nx_iy_i$.
In our setting, we are also given a function $f:\R\rightarrow\R$ and a vector $x\in\R^n$ where the storage of $x$ is free, but not directly given $y$.
The \emph{$f$-vector} product is defined as $\langle x, f(y)\rangle$, where $f$ is applied to $y$ coordinate-wisely.
The updates to $y$ is a stream, i.e., we observe a sequence of integer pairs $(z_t, \Delta_t)$ for $t=1, 2, \ldots, m$, 
where each $z_t\in [n]$ and $\Delta_t\in\{-1,1\}$.
Thus, we initialize $y$ as a $y^{(0)}\gets0$, a zero-vector,
and at time $t$, the update to $y$ is described by
$
y^{(t)}\gets y^{(t-1)} +  \Delta_{z_t}\cdot e_{z_t}
$ 
where $e_{z_t}$ is the standard unit vector with only the $z_t$-th coordinate non-zero.
Our goal is to approximate $\langle x, f(y)\rangle$ without storing $y$, where $x$ is given to the algorithm without storage cost.  Formally, we define the following problem. 
	
\begin{problem}[approximate transformed vector and vector inner product]
Given a fixed vector $x$ and function $f:\R\rightarrow \R$, design an algorithm that makes a single pass over an update stream of a vector $y$, output an approximated value of $\langle f(y) , x \rangle$ with high probability. We require the algorithm to use as small space as possible (excluding the space of $x$).	
\end{problem}
We note that a na\"ive algorithm would be storing the vector $y$ as a whole. 
However such an algorithm is not feasible when $n$ is large or the demand of computing such inner products is too high (e.g., in our matrix applications for computing $Z=f(A)B\in \R^{n\times k}$, each entry of $Z$ is an inner product. 
If each inner product requires space $n$, then final space can be $O(n^2k)$ which is prohibitively high.). 
In Section~\ref{sec:sketch_logsum} below, we design an algorithm that accomplish this task for function $f(y)=\log(|y|+1)$, which only uses $\wt{O}(1)$ bits of memory.
In Section~\ref{sec:more general function}, we present a general framework that works for a general family of functions $f$ with nearly optimal space complexity.

\subsection{Sketch $\log(|\cdot |+1)$-Vector Product} \label{sec:sketch_logsum}

{\small
\begin{algorithm}[!hbt]
		\begin{algorithmic}[1]
		\State {\bf data structure} \textsc{LogSum} \Comment{Theorem~\ref{thm:logsum}}
		\Procedure{Initialize}{$x$}
			\State$\gamma\leftarrow\epsilon^{-2}\poly(\log n/\delta )$\label{eq:alg_gamma} 
			\State $t\leftarrow \Theta(\log n)$, $p_j \leftarrow 2^{-j}\cdot \gamma, \forall j \in [t]$
			\For{$j = 1 \to t$}
			\State Sample a $\log n$-wise independent hash function $h_j : [n] \rightarrow \{0,1\}$ such that 
			$
			\forall i
			\in[n]: \Pr[ h_j (i) = 1] = \min(  p_j, 1).
			$
			\State Sample a K-set structure \textsc{KSet}$_j$ with error parameter $\Theta(\delta/t)$ and memory budget $\epsilon^{-2}\poly(\log n/\delta)$
			\EndFor
		\EndProcedure
		\Procedure{\text{Update}}{$a$} \Comment{$a \in [n]$}
			\For{$j=1 \to t$}
			\If{$h_j(a) = 1$ and $x_{a}\neq 0$}
			\State \textsc{KSet}$_j$.update($a$)
			\EndIf
			\EndFor
			\EndProcedure
			\Procedure{\text{Query}}{$ $}
			\State Pick the largest $j$ such that \textsc{KSet}$_j$ does not return ``Fail''
			\State Let $v$ be the output of \textsc{KSet}$_j$, denote $S_j=\supp(v)$ 
			\State \Return $2^j \sum_{i \in S_j} x_i \log (|v_i|+1) $ 
			\EndProcedure
		\State {\bf end data structure}
		\end{algorithmic}\caption{}\label{alg:logsum} 
\end{algorithm}\vspace{-1mm}
}

Recall that, when $f(\cdot) = \log(|\cdot|+1)$, we are designing an algorithm for computing the inner product $\langle \log(|y|+1), x\rangle$, where $x,y\in \R^{n}$ are two vectors,  $x$ is given to the algorithm for free and $y$ is under updating.
Our full algorithm is Algorithm~\ref{alg:logsum}, which is composed of 3 sub-procedures: 
procedure \textsc{Initialize} is called on initialization with given vector $x$,
procedure \textsc{Update} is called when we go over the update stream of the vector $y$,
and procedure \textsc{Query} is called at the end to report the answer.
The detailed analysis of Algorithm~\ref{alg:logsum},\ can be found in Appendix \ref{sec:additional results}. 
We here sketch the high level ideas for how it works.
For ease of representation, we consider $x$ has no zero coordinates, since otherwise we can simply ignore these coordinates and change our universe $[n]$ to $\supp(x)$ accordingly.
Our algorithm is originated from \cite{bo10a} but it is much simplified in this paper.
From a high level, our algorithm can be viewed as an $\ell_0$-sampler, namely, sample uniformly at random from the support of an updating vector $y$.
Note that the support of $y$ is changing over time.
Thus it is non-trivial to maintain a uniform sample while using only small space.
We also note that it is necessary to sample coordinates from the support of $y$, since otherwise we can always construct worst-case examples for algorithms that sample coordinates uniformly from $[n]$.

We design our algorithm thus by maintaining independently $\Theta(\log n)$ many sub-vectors of the vector $y$.
Each sub-vector is generated by sampling a set of coordinates uniformly from $[n]$ with geometrically decreasing probabilities. 
For instance, in our algorithm, we first generate $\Theta(\log n)$ many hash functions, each defines a set $S_j\subset[n]$.
For each $i\in[n]$, we demand that $i\in S_j$ with probability $2^{-j}$.
Thus if the size of the support of $y$ is of order $\Theta(2^j)$, then we are expected to sample $\Theta(1)$ samples of $y$ using the set $S_j$.
We now describe how to maintain these sampled coordinates in memory.
For convinience we assume $\gamma=1$ in line \ref{eq:alg_gamma} in Algorithm~\ref{alg:logsum}.

For the case of insertion-only stream (once a coordinate of $y$ becomes larger than $0$, it stays so),
maintaining the sub-vector $y_{S_j}$ is a trivial task since the number of coordinates of $y_{S_j}$ is expected to be $O(1)$.
However, for $j'\le j$, the sub-vectors $y_{S_{j'}}$s contain too many coordinates.
We handle this quite straightforwardly: if any of them exceeds our memory budget, we just ignore them.
For the case of general stream, in which coordinates can be $0$ even they were non-zero at some time-point.
We will be using the K-set data structure presented in \cite{g07}. This data structure supports insertion and deletion of data points and can maintain the samples only if the number of final samples is under the memory budget.
The formal guarantee of the $K$-set data structure presented in Theorem~\ref{lem:k-set}.

Suppose now we have collected sufficiently many samples from the support of the vector $y$. 
Suppose the set of samples is collected using set $S_j$.
We can have an empirical estimator for the inner product as $2^{j}\sum_{i\in S_j} x_i\log(|y_i|+1)$.
Notice that this estimator is unbiased.
Also since the variance of the estimator is bounded by  
\begin{align*}
& ~ \sum_{i}2^jx_i^2\log^2(|y_i|+1) \\
= & ~ O(1) \cdot \|x\|_\infty^2\cdot\sum_{i}\log^2 (|y_i|+1) \cdot\log^2 m,
\end{align*}
 where $m$ is the length of the stream and is usually assumed to be of oder $\poly(n)$, thus we only need $\poly\log n$ samples to obtain an accurate estimation.

We summarize the main guarantee in the following theorem, while the formal proof can be found in Section~\ref{sec:additional results}.

\begin{theorem}[approximate inner product of transformed vector and vector]\label{thm:logsum}
Suppose vector $x\in \R^n$ is given without memory cost. There exists a streaming algorithm (data structure \textsc{LogSum} in Algorithm~\ref{alg:logsum}) that makes a single pass over the stream updates  to a vector $y\in \R^n$ and outputs $Z\in \mathbb{R}$, 
	such that, with probability at least $1-\delta$,
\begin{align*}
		|Z - \langle x, \log(|y|+1)\rangle| \le \epsilon\cdot \|x\|_\infty \cdot\sum_{i=1}^n\log(|y_i|+1).
\end{align*}
The algorithm uses space $O(\epsilon^{-2} \poly(\log ( n/\delta)))$ (excluding the space of $x$) has a $\poly(\log n,1/\epsilon)$ query time.
\end{theorem}
\begin{remark}
We also note that our algorithm naturally works for $f(y):=\log^c(|y|+1)$ for any constant $c$.
To modify our algorithm, we only need to keep slightly larger space and change the final estimation to be $2^j \sum_{i \in S_j} x_i \log^c (|v_i|+1)$. It also enjoys the same relative error guarantee in Theorem~\ref{thm:logsum}.
\end{remark}

\subsection{From Vector Product Sketch to Matrix Product Sketch} 
With the $f$-inner product sketch tools established, we are now ready to present the result for sketching the matrix product, $Z=f(A) B$.
Notice that each entry $Z_{i,j} := \langle f(A_{i}), B_j\rangle$ is an inner product.

Thus our algorithm for the matrix sketch is simply maintaining an $f$-inner product sketch for each $Z_{i,j}$.
In our algorithm, we assume that matrix $B$ is  given to the algorithm for free.
Thus, if $B\in\R^{n\times k}$ for some $k\ll n$, we only need to keep up to $\wt{O}(nk)$ vector product sketches, which cost in total $\wt{O}(nk)$ words of space.
For the ease of representation, we present our guarantee for matrix product for $f(z):=\log^c(|z|+1)$ for some $c$ or for $f(z) = z^{p}$ for $0\le p\le 2$, and for matrix $B\in \{-1, 0, 1\}^{n\times k}$. 
Our results can be generalized to a more general set of functions  and matrix $B$ using the results presented in Section~\ref{sec:more general function}. 
The proof of the following theorem is a straightforward application of Theorem~\ref{thm:logsum} and \ref{thm:sqrt-sketch}.
\begin{theorem}[approximate each coordinate of the transformed matrix]
	Given a matrix $B\in \{-1, 0, 1\}^{n\times k}$, and a function $f(x):=\log^c(|x|+1)$ for some $c$ or $f(x):=|x|^p$ for some $0\le p\le2$, then
	there exists a one-pass streaming algorithm that makes a single pass over the stream updates to an underlying matrix $A\in \R^n$ 
	and outputs a matrix $\wh{Z}$, such that, with probability at least $1-\delta$, for all $i,j$,
	\begin{align*}
	|\wh{Z}_{i,j} - Z_{i,j}| \le \epsilon \sum_{j'=1}^n f(|A_{i,j'}|).
	\end{align*}
	The algorithm uses space $\epsilon^{-2}nk\poly(\log ( n/\delta))$  and has an $nk\poly(\log n, 1/\epsilon)$ query time. 
\end{theorem}

\begin{remark}\label{remark:norm}
We note that our sketch in the last theorem can be easily used to approximate the $2$-norm of each row of the matrix $f(A)$.
In this case, we simply choose $B\in \R^{n\times 1}$ as the all-$1$ vector and change $f(\cdot)$ to be $f^2(\cdot)$. 
For $f(x) =\poly\log(|x|+1)$ or $f(x) = |x|^p$ with $0\le p\le 1$, it can be easily verify that our output is a $(1\pm\epsilon)$ approximation to $f^2(A) \cdot \boldmath{1}$, hence the approximation of $2$-norm squared of each row of $f(A)$.
\end{remark}


\section{Application to Low Rank Approximation} \label{sec:lowrank}

\begin{algorithm*}[!t]

\caption{Low rank approximation of $M = \log(|A|+1)$} \label{alg:main_lowrank}
\begin{algorithmic}[1]
\Procedure{\textsc{LowRankApprox}}{$A,k,\epsilon$} \Comment{Theorem~\ref{thm:main_lowrank}}
\State $s \leftarrow O(k \log k)$
\State $d_1 \leftarrow O(k \log^2 k)$
\State $d_2 \leftarrow O(k/\epsilon)$
\State $\eta \leftarrow O( \epsilon \sqrt{d_1} + \epsilon^2 d_1 )$
\State \Comment{Step 1 : Sampling according to generalized leverage scores of $M$}
\State Let $S$ be the CountSketch (SparesJL) matrix of size $s \times n$ \Comment{ Appendix ~\ref{sec:def_count_sketch_gaussian}}   
\State Let $S_+$ and $S_-$ be its positive and negative parts of $S$. 
\State $R \leftarrow [ S_+ ; S_- ]$
\State $\tilde{E} \leftarrow \textsc{LogSum}(RM)$ \Comment{ $\| \wt{E}_i \|_2^2 = (1\pm \epsilon) \| ( RM )_i \|_2 , \forall i $}
\State Sample a set $P$ of $d_1$ columns of $M$ according to the leverage score of $\wt{E}$.\label{alg:sample p} \Comment{Definition~\ref{def:leverage_score_sampling}}
\State \Comment{Step 2 : Adaptive sampling}
\State$[Q_p, \cdot ] \leftarrow \textsc{QRFactorization}(P)$ \label{alg:qrdecomposition} \Comment{$Q_p$ is the basis vectors for $P$}
\State $\tilde\Gamma \leftarrow \textsc{LogSum} ( Q_p^\top M )$ \Comment{ $\| \wt{\Gamma}_i \|_2^2 = (1\pm \epsilon) \| ( Q_p^\top M )_i \|_2^2, \forall i$ }
\State $\tilde{z} \leftarrow \textsc{LogSum} (M)$ \Comment{ $\tilde{z}_i = (1 \pm \epsilon)\|M_i\|_2^2,\forall i$ }
\State $\tilde s_i \leftarrow \tilde{z}_i - \|\tilde\Gamma_i\|_2^2$
\State Sample a set $\tilde{Y}$ of $d_2$ columns from $M$ according to $ p_i = \max ( \tilde s_i, \eta \tilde{z}_i ) $
\State $Y \leftarrow \tilde{Y} \cup P$
\State \Comment{Step 3 : Computing approximation solutions}
\State $[Q_y, \cdot ] \leftarrow \textsc{QRFactorization} ( Y )$ \Comment{$Q_y$ is the basis vectors for $Y$}
\State  $\tilde\Pi \leftarrow \textsc{LogSum} (Q_y^\top M)$ \Comment{ $ \| \wt{\Pi}_i \|_2^2 = (1 \pm \epsilon^2 ) \| (Q_y^\top M)_i \|_2^2 , \forall i$}
\State Compute the top $k$ singular vectors $\tilde{W}$ of $\tilde\Pi$
\State $L \leftarrow Q_y\tilde{W}$
\State \Return $L$
\EndProcedure
\end{algorithmic}
\end{algorithm*}

This section considers the concrete application of rank-$k$ approximation for $M$ where $M_{i,j}=\log(|A_{i,j}|+1)$, i.e., finding $k$ orthonormal vectors $L$ such that $\|M - LL^\top M\|_F$ is minimized. Our algorithm for rank-$k$ approximation is presented in Algorithm \ref{alg:main_lowrank}. Low rank approximation for other functions $f$ follows the same algorithm and  similar analysis.

There exists a large body of work for low rank approximation (see, e.g.,~\cite{halko2011finding,DriMagMahWoo12,w14,cw13,mm13,nn13,cw15focs,rsw16,swz17,cgklpw17,swz18,bw18,kprw19,swz19c,swz19,swz19b,song19,bbbklw19,djssw19,bcw19,ivww19,bwz19} and references therein) but most of them are designed for the case without transformation and thus cannot be directly applied. As mentioned in previous sections, if an algorithm only accesses the transformed matrix via a matrix product, plugging in our sketching method leads to a suitable algorithm.
We design an algorithm that applies generalized leverage score sampling approach~\cite{DriMagMahWoo12,balcan2016communication} for low-rank approximation.
Leverage score sampling is a non-oblivious sketching technique that is widely used in numerical linear algebra and has been successfully
applied to speed up different problems such as linear regression \cite{cw13,psw17,akklns17,swz19,dswy19}, row sampling \cite{ss11,lmp13}, spectral approximation \cite{clmmps15}, low rank approximation \cite{bw14,swz17,swz19}, cutting plane methods \cite{v89,lsw15,jlsw20}, linear programming \cite{blss20}, computing John Ellipsoid \cite{ccly19}. From the perspective of graph problems, leverage score is closely related to random spanning tree \cite{s18,ks18}, graph sparsification and Laplacian system solver \cite{st04,ss11,bss12}. 
Readers may refer to Appendix \ref{sec:leverage_score} for more detailed discussion on leverage score sampling.

On a high level, we would like to sample matrix $M\in \R^{n\times n}$ according to its leverage scores. It turns out it is sufficient to use the leverage scores of $SM$ where $S$ is a sketching matrix. We apply Algorithm~\ref{alg:logsum} to do so and obtain the sampled set $P$ (Step 1). We then apply the technique of adaptive sampling to refine the sampling and obtain $Y$ (Step 2) so that we have better control over the rank, and finally compute the solution using $Y$ by taking projection and computing singular vectors (Step 3). Detailed description and analysis of Algorithm \ref{alg:main_lowrank} can be found in Appendix \ref{sec:proof_sv_1}. Overall we have the following guarantee.

%

\begin{theorem}[low-rank approximation]\label{thm:main_lowrank}
For any parameter $\epsilon \in (0,1)$ and integer $k \geq 1$, 
there is an algorithm (procedure \textsc{LowRankApprox} in Algorithm~\ref{alg:main_lowrank}) that runs in $\wt{O}(n) \cdot k^3\cdot \poly( 1/\epsilon )$ time, takes $\tilde{O}(n) \cdot k^3/\epsilon^2$ spaces, and outputs a matrix $L\in \R^{n\times k}$ such that
\begin{align*}
	\| LL^\top M - M \|_F^2  \leq & ~ 10 \cdot \| M - [M]_k \|_F^2
	  \\
	  & ~ + O\left( \frac{\epsilon^2}{k^3\log^5 k} \right) \cdot \| M \|_{1,2}^2,
\end{align*}
holds with probability at least $9/10$, where $\| M \|_{1,2} = ( \sum_j \| M_{*,j} \|_1^2 )^{1/2}$. 
\end{theorem}

For a large $n$ and fixed $\epsilon$, our algorithm uses much less space than storing the full matrix. 
Note that our algorithm still needs to make several passes over the stream of updates. Whether there exists a one-pass algorithm is still an open problem, and is left for future work. 

\section{Experiments} \label{sec:exp}

\begin{figure*}[!th]
\centering
\subfloat[\LOGDATA, $n=10^4$]{\includegraphics[height=0.25\linewidth]{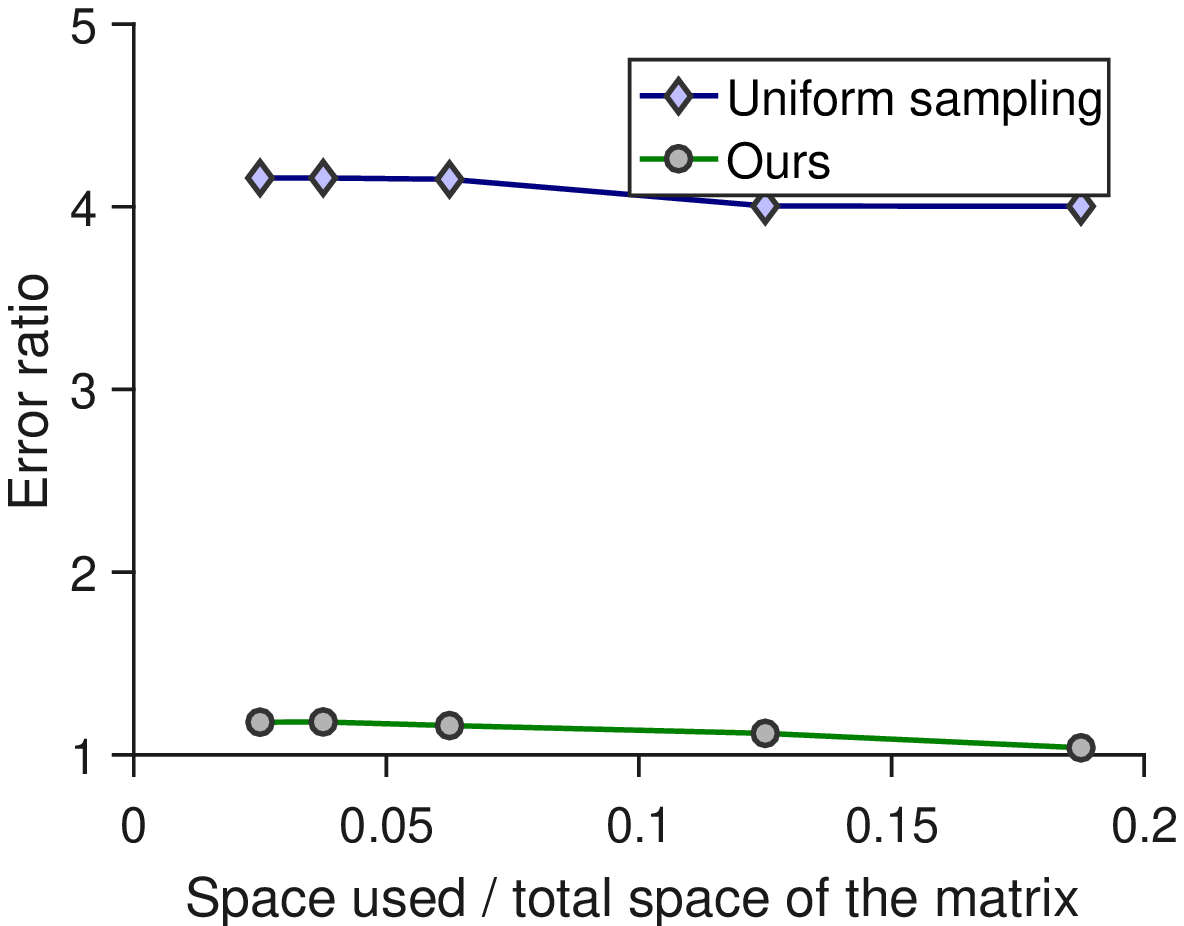}}~
\subfloat[\LOGDATA, $n=3\cdot 10^4$]{\includegraphics[height=0.25\linewidth]{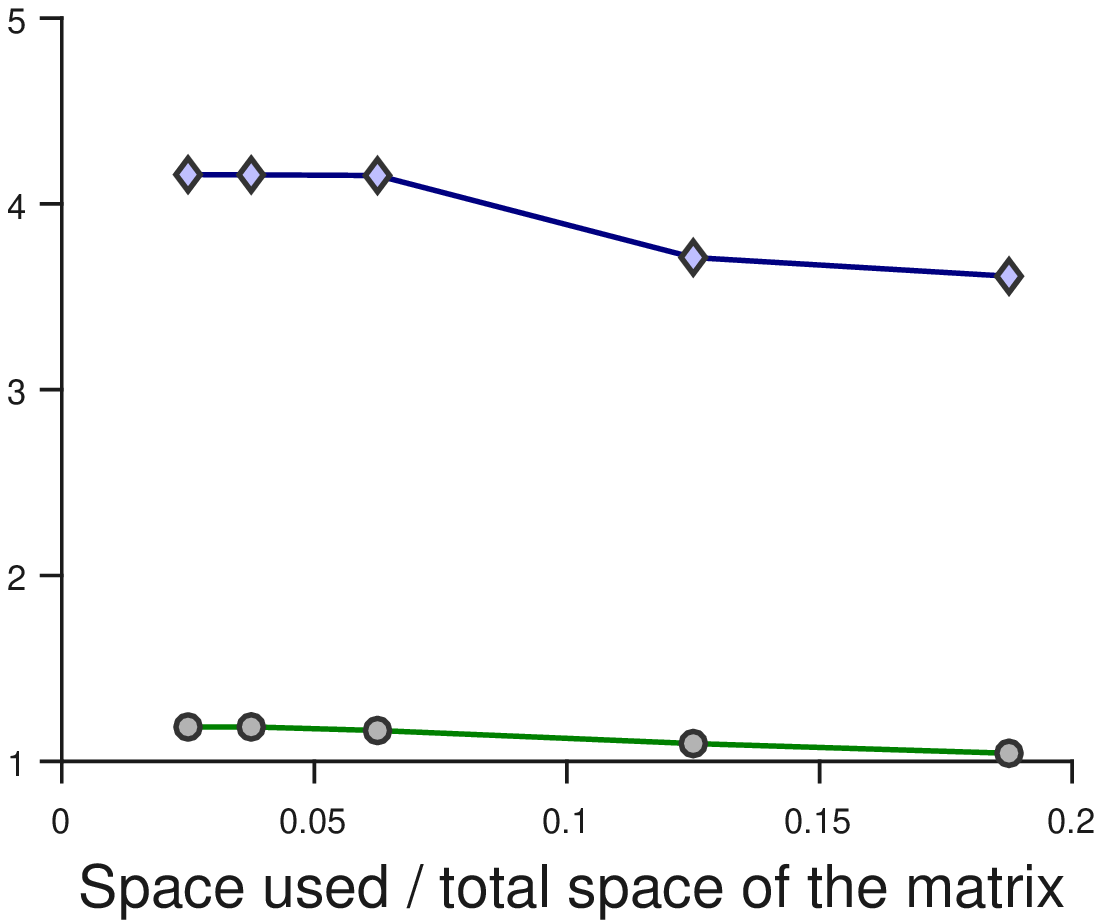}}~
\subfloat[\LOGDATA, $n=5\cdot 10^4$]{\includegraphics[height=0.25\linewidth]{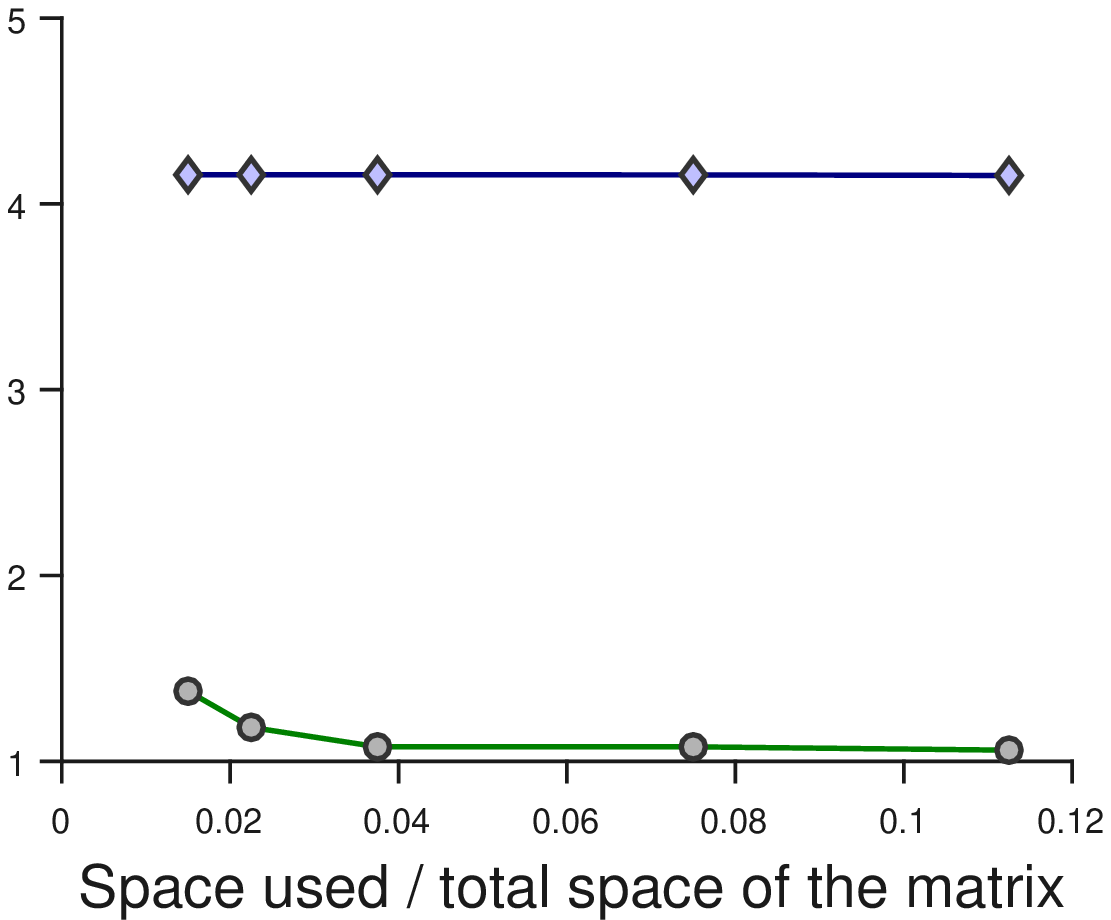}}
\\
\subfloat[Real data, $n=10^4$]{\includegraphics[height=0.25\linewidth]{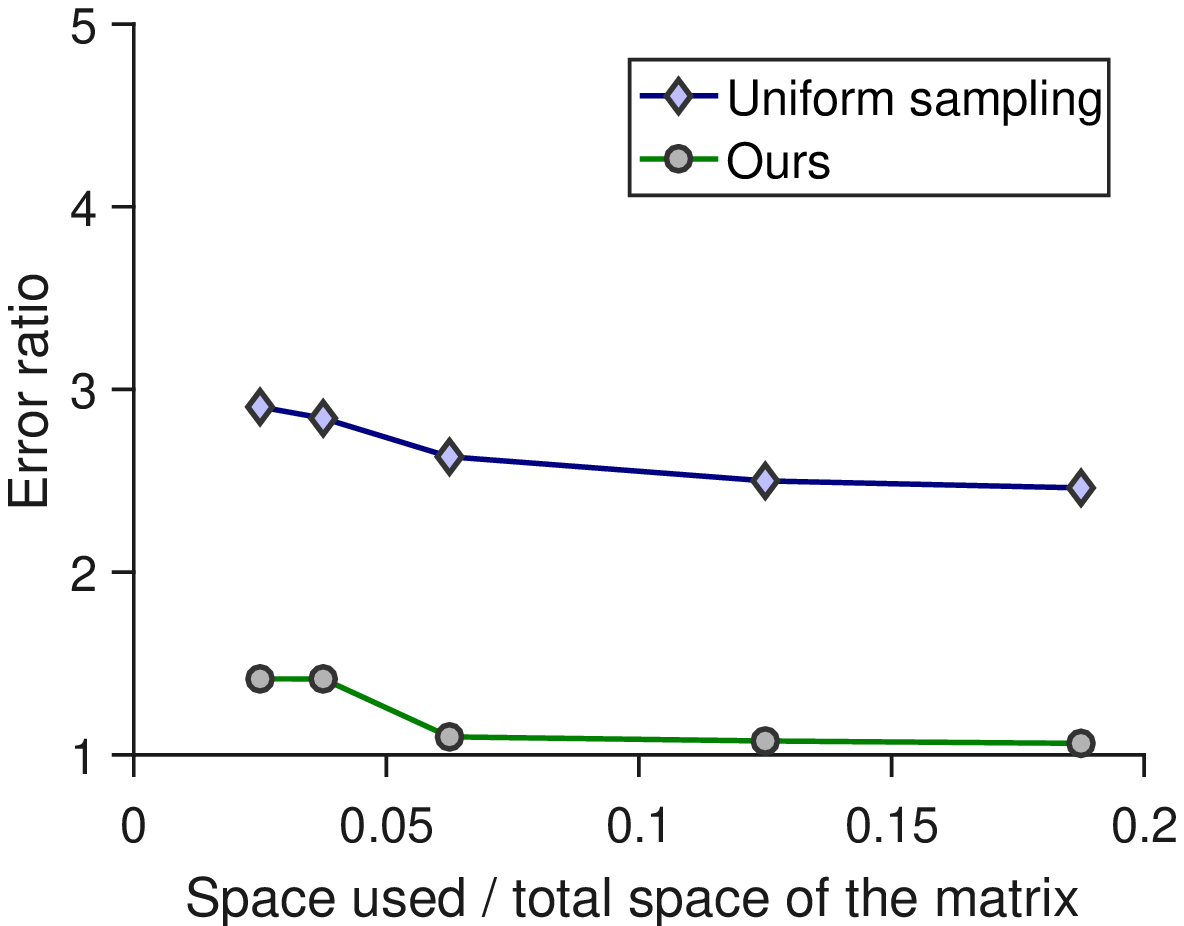}}~
\subfloat[Real data, $n=3\cdot 10^4$]{\includegraphics[height=0.25\linewidth]{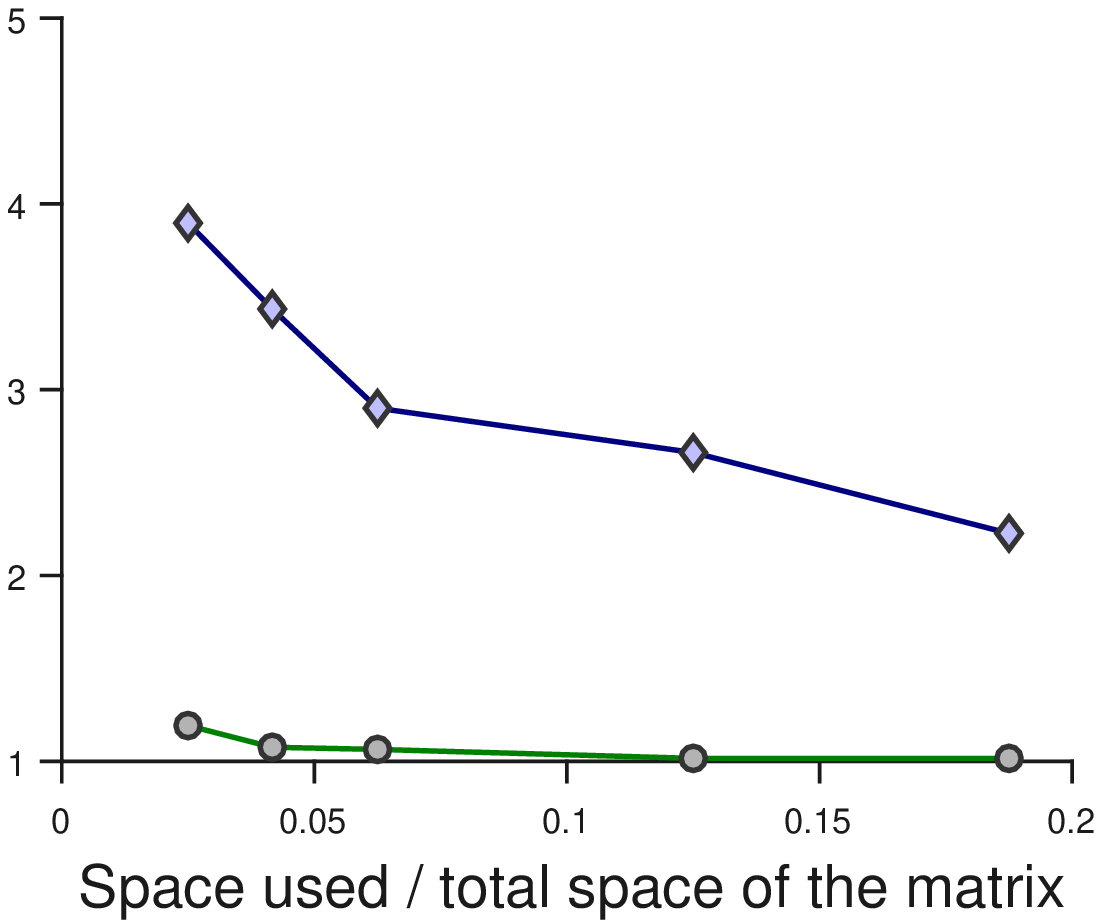}}~
\subfloat[Real data, $n=5\cdot 10^4$]{\includegraphics[height=0.25\linewidth]{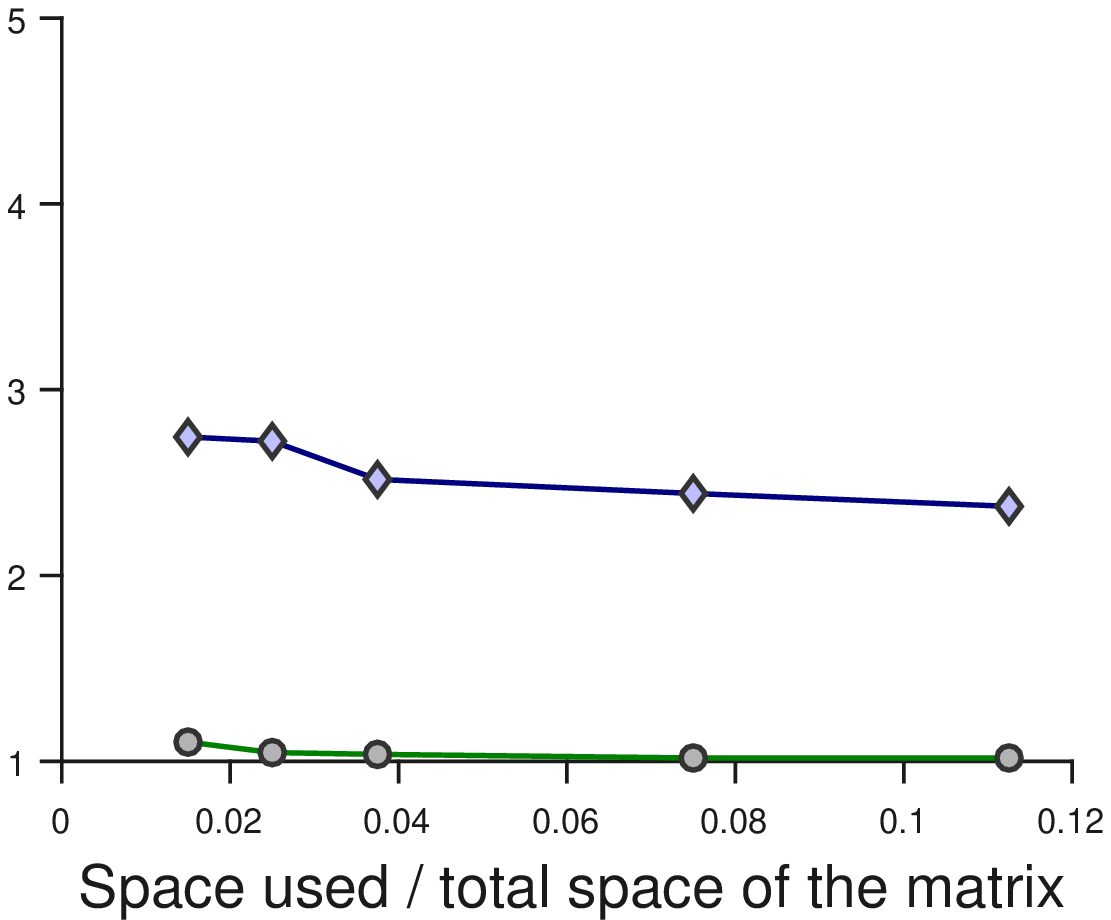}}
\caption{ Error ratios on the synthetic data (top row) and the real data (bottom row). The $x$-axis is the ratio between the amount of space used by the algorithms and the total amount of space occupied by the data matrix. The $y$-axis is the ratio between the error of the solutions output by the algorithms and the optimal error. 
}
\label{fig:syn}
\end{figure*}


To demonstrate the advantage of our proposed method, we complement the theoretical analysis with empirical study on synthetic and real data. We consider the low rank approximation task with $f(x)=\log(|x|+1)$.
We adjust the constant factors in the amount of space used by our method and compare the errors of the obtained solutions. 
In the appendix, we describe more experimental details.
We also provide additional experiments in the appendix to show that the method also works for $f(x)=\sqrt{|x|}$.

We furthre demonstrate the robustness of the parameter selections in the algorithm.

\textbf{Setup.}
Given a data stream in the form of $(i_t,j_t,\delta_t)$, we use the algorithm in Section~\ref{sec:lowrank} to compute the top $k=10$ singular vectors $L$, and then compare the error of this solution to the error of the optimal solution (i.e., the true top $k$ singular vectors). 
Let $A$ denote the accumulated matrix, $M=f(A)$ denote the transformed one, and $U$ denote the top $k$ singular vectors of $M$. 
Then the evaluation criterion is 
\begin{align*}
  \text{error-ratio}(L) = \|M - LL^\top M\|_F / \|M-UU^\top M\|_F.
\end{align*}
Clearly, the error ratio is at least $1$, and a value closer to $1$ means a better solution.  

Besides demonstrating the effectiveness, we also exam the tradeoff between the solution quality and the space used. Recall that there are constant parameters in the sketching methods controlling the amount of space used. We vary its value, and set the parameters in other steps of our algorithm so that the amount of space used is dominated by that of the sketch. We then plot how the error ratios change with the amount of space used.
The plotted results are averages of 5 runs; the variances are too small to plot. 
Finally, we also report the results of a baseline method: uniformly at random sample a subset $T$ of columns from $A$, and then compute the top $k$ singular vectors of $f(T)$. The space occupied by the columns sampled is similar to the space required by our algorithm for fair comparison. 
We choose uniform sampling as baseline because to the best of the authors' knowledge,
our algorithm is the first one to deal with low-rank approximation on transformed matrix in the stream setting,
and we are not aware of any other non-trivial algorithm working in this setting.

\subsection{Synthetic Data}

\paragraph{Data Generation.}

The data sets \LOGDATA~are generated as follows. First generate a matrix $M$ of $n\times n$ where the entries are i.i.d.\ Gaussians. To break the symmetry of the columns, we scale the norm of the $i$-th column to $4/i$. Finally, we generate matrix $A$ with $A_{ij}=\exp(M_{ij})-1$. Each entry $A_{ij}$ is divided into equally into $5$ updates $(i,j,A_{ij}/5)$, and all the updates arrive in a arbitrary order. The size $n$ can be $10000$, $30000$, and $50000$.

\paragraph{Parameter Setting.}

In our algorithm for low rank approximation, an FJLT matrix $S$ is used~\cite{a03,ac06}. For the sketching subroutine, instead of specifying the desired $\epsilon$, we directly set the size of the data structure (line 19 in \textsc{LogSum}), so as to exam the tradeoff between space and accuracy. 
We set $m_c=m_s=m_a$ and set their value so that the space used is at most that used by the sketch method. 

\paragraph{Results.}

Figure~\ref{fig:syn} top row shows the results on the synthetic data. 
In general, the error ratio of our method is much better than that of the uniform sampling baseline: ours is close to 1 while that of uniform sampling is about 4. 
It also shows that our method can greatly reduce the amount of space needed, e.g., by orders of magnitude, but still preserve a good solution. 
This advantage is more significant on larger data sets.  For example, when $n=50000$, to obtain $5\%$  error over the optimum solution, we only needs space corresponding to $5\%$ of the size of the matrix.

\subsection{Real Data}

We experiment our method on the real world data from NLP applications, which are the motivating examples for our approach.
Our method with $f(x)=\log(|x|+1)$ is used. The parameters are set in a similar way as for the synthetic data.

\paragraph{Data Collection.}

The data set is the entire Wikipedia corpus~\cite{enwiki} consisting of about 3 billion tokens. Details can be found in the appendix and only a brief description is provided here. The matrix to be factorized is $M$ 
with $M_{ij} = p_j\log (\frac{N_{ij} N }{N_i N_j}+1)$ where 
$N_{ij}$ is the number of times words $i$ and $j$ co-occur in a window of size $10$, $N_i$ is the number of times word $i$ appears, $N$ is the total number of words in the corpus, and $p_j$ is a weighting factor depending on $N_j$ (putting larger weights on more frequent words). 
Note that $N_i$'s and $N$ can be computed easily, so essentially the only dynamically update part is $\log N_{ij}$. The data stream is generated by considering each window of size 10 along the sentences in the corpus and collecting the co-occurrence counts of the word pairs in that window.
We consider the matrix for the most frequent $n$ words, where $n=10000$, $30000$, and $50000$.

\paragraph{Results.}

Figure~\ref{fig:syn} bottom row shows the results on the real data. The observations are similar to those on the synthetic data: the errors of our method are much smaller than the baseline, and are close to the optimum. 
 These results again demonstrate the accuracy and space efficiency of our methods.
\section{Conclusions} \label{sec:conclusion}

We considered the setting where a large matrix is updated by a data stream and the learning tasks is performed on an element-wise transformation of the matrix. 
We proposed a method for computing the product of its element-wise transformation with another given matrix.
For a large family of transformations, 
our method only needs a single pass over the data and provable guarantees on the error.
Our method uses much smaller space than directly storing the matrix.  
Our approach can be used as a building block for many learning tasks.
We provided a concrete application for low-rank approximation with theoretical analysis and empirical verification, showing the effectiveness of this approach.


\newpage

\bibliographystyle{alpha}
\bibliography{ref}

\onecolumn
\appendix

\section{Preliminaries}

\subsection{CountSketch and Gaussian Transforms}\label{sec:def_count_sketch_gaussian}
\begin{definition}[Sparse embedding matrix or CountSketch transform]\label{def:count_sketch_transform}
A CountSketch transform is defined to be $\Pi=  \Phi D\in \mathbb{R}^{m\times n}$. Here,  $D$ is an $n\times n$ random diagonal matrix with each diagonal entry independently chosen to be $+1$ or $-1$ with equal probability, and $\Phi\in \{0,1\}^{m\times n}$ is an $m\times n$ binary matrix with $\Phi_{h(i),i}=1$ and all remaining entries $0$, where $h:[n]\rightarrow [m]$ is a random map such that for each $i\in [n]$, $h(i) = j$ with probability $1/m$ for each $j \in [m]$. For any matrix $A\in \mathbb{R}^{n\times d}$, $\Pi A$ can be computed in $O(\nnz(A))$ time. 
\end{definition}

\begin{definition}[Gaussian matrix or Gaussian transform]\label{def:gaussian_transform}
Let $S=\frac 1{\sqrt{m}}\cdot G \in \mathbb{R}^{m\times n}$ where  each entry of $G\in \mathbb{R}^{m\times n}$ is chosen independently from the standard Gaussian distribution. For any matrix $A\in \mathbb{R}^{n\times d}$, $SA$ can be computed in $O(m \cdot \nnz(A))$ time. 
\end{definition}

We can combine CountSketch and Gaussian
transforms to achieve the following:
\begin{definition}[CountSketch + Gaussian transform]\label{def:fast_gaussian_transform}
Let $S' = S \Pi$, where $\Pi\in \mathbb{R}^{t\times n}$ is the CountSketch transform (defined in Definition~\ref{def:count_sketch_transform}) and $S\in \mathbb{R}^{m \times t}$ is the Gaussian transform (defined in Definition~\ref{def:gaussian_transform}). For any matrix $A\in \mathbb{R}^{n\times d}$, $S'A$ can be computed in $O(\nnz(A) + dtm^{\omega-2})$ time, where $\omega$ is the matrix multiplication exponent.
\end{definition}

\subsection{Pythagorean Theorem, matrix form}
Here we state a Pythagorean Theorem for matrices.
\begin{theorem}[Pythagorean Theorem]\label{thm:pythagorean}
For any integers $m,n>0$ and matrices $A,B\in \R^{m\times n}$,
if $\Tr[A^\top B]=0$,
then 
\begin{align*}
\|A+B\|_F^2=\|A\|_F^2+\|B\|_F^2
\end{align*}
\end{theorem}

\subsection{Adaptive Sampling}
We described a $t$-round adaptive sampling algorithm. The algorithm is originally proposed in \cite{drvg06}. We will use $\pi_V(A)$ to denote the matrix obtained by projecting each row of $A$ onto a linear subspace $V$. If $V$ is spanned by a subset $S$ of rows, we denote the projection of $A$ onto $V$ by $\pi_{\text{span}(S)}(A)$. We use $\pi_{\text{span}(S),k}(A)$ for the best $\rank$-$k$ approximation to $A$ whose rows lie in $\text{span}(S)$.

\begin{itemize}
\item Start with a linear subspace $V$. Let $E_0 = A - \pi_V (A)$ and $S = \emptyset$

\item For $j=1$ to $t$, do
\begin{itemize}
	\item Pick a sample $S_j$ of $s_j$ rows of $A$ independently from the following distribution : row $i$ is picked with probability $P_i^{(j-1)} \geq c \frac{ \| E_{j-1}^{(i)} \|_2^2 }{ \| E_{j-1} \|_F^2 }$
	\item $S = S \cup S_j$
	\item $E_j = A - \pi_{\text{span}}(V \cup S)(A)$.
\end{itemize}

\end{itemize}

\begin{theorem}[\cite{drvg06}, see also Theorem 3 in \cite{dv06}]\label{thm:theorem_3_in_dv06}
After one round of the adaptive sampling procedure described above,
\begin{align*}
\E_{S_1} [ \| A - \pi_{\text{span}(V \cup S_1) , k} (A) \|_F^2 ] \leq \| A - A_k \|_F^2 + \frac{k}{c s_1} \| E_0 \|_F^2.
\end{align*}
\end{theorem}

\section{Additional Results for Sketching $f$-Matrix Product}
\label{sec:additional results}

\subsection{Proofs of Sketch $\log(|\cdot|+1)$-Vector Product}

\begin{theorem}[\cite{g07}, K-Set]
	\label{lem:k-set}
	There exists a data structure supports updates of the form $(i, \Delta)$ to a vector $v\in\R^n$, where $i\in [n]$ and $\Delta\in \{-1,1\}$, and supports a query operation at any time. The algorithm either returns the current vector $v\in \R^n$  or ``Fail''. 
	If the  $\supp(v)\le k$, then the data structure returns ``Fail'' with probability at most $\delta\in(0,1)$.
	The algorithm uses space $O[k\log n\log (k/\delta)]$ bits.
\end{theorem}

\begin{proof}[Proof of Theorem~\ref{thm:logsum}]


	Firstly, in the algorithm, for the level $j$, we sample the universe with probability $p_j = \min(\frac{\epsilon^{-2}\poly(\log n/\delta)}{2^j}, 1)$.
	Suppose the true support of $x$ satisfies $|\supp(x)| = \Theta(2^{j})$.
	We argue that with high probability, there exists an $j^*\ge j$ such that \textsc{KSet}$_{j^*}$ succeeds. 
	To show this, it is suffice to show that \textsc{KSet}$_{j}$ succeeds with high probability.
	By Chernoff bound, with probability at least $1-\Theta(\delta)$, the number of coordinates sampled in level $j$ is $\Theta(\epsilon^{-2}\poly(\log n/\delta))$. 
	By Theorem~\ref{lem:k-set}, the \textsc{KSet}$_j$ instance succeeds to return the sampled sub vector with probability at least $1-O(\delta)$. 
	Since the coordinates sampled in \textsc{KSet}$_{j^*}$ is with probability at least $p_j$, we can bound the variance of unbiased estimator by 
	\allowdisplaybreaks
	\begin{align*}
	&\E_{S_j}\left[\left(\sum_{i\in S_j}2^jx_i\log(|y_i|+1)\right)^2\right]-\left(\E_{S_j}\left[\sum_{i\in S_j}2^jx_i\log(|y_i|+1)\right]\right)^2\\
	=&\E_{S_j}\left[\sum_{i\in S_j}\sum_{k\in S_j}(2^jx_i\log(|y_i|+1))(2^jx_k\log(|y_k|+1))\right]-\left( \sum_{i=1}^n\Pr[i\in S_j]2^{j}x_i\log(|y_i|+1)\right)^2\\
	=&\sum_{i=1}^n\sum_{k=1}^n\Pr[i\in S_j, k\in S_j](2^jx_i\log(|y_i|+1))(2^jx_k\log(|y_k|+1))-\left( \sum_{i=1}^n\Pr[i\in S_j]2^{j}x_i\log(|y_i|+1)\right)^2\\
	=&\sum_{i=1}^n  \Pr[i\in S_j]2^{2j}x_i^2\log^2(|y_i|+1)-\sum_{i=1}^n  \Pr[i\in S_j]^22^{2j}x_i^2\log^2((|y_i|+1))\\
	\leq &\sum_{i=1}^n  2^{j}x_i^2\log^2(|y_i|+1)\\
	\leq &\max_{i\in [n]}\left(2^{j}x_i^2\log(|y_i|+1)\right)\cdot \sum_{i=1}^n\log (|y_i|+1)  \\
	\leq &\max_{i\in [n]} x_i^2\cdot \max_{i\in [n]}\left(2^{j}\log(|y_i|+1)\right)\cdot \sum_{i=1}^n\log (|y_i|+1)  \\
	\leq & \|x\|_{\infty}^2\cdot (\sum_{i=1}^n\log (|y_i|+1))^2\cdot \log m
	\end{align*}
	where the first step uses the fact \begin{align*}\E_{S_j}\left[\sum_{i\in S_j}2^jx_i\log(|y_i|+1)\right]=\sum_{i=1}^n\Pr[i\in S_j]2^{j}x_i\log(|y_i|+1),\end{align*}
	the second step expands the square,
	the fourth step uses $\Pr[i\in S_j]=2^{-j}$,
	the fifth step uses the fact that $\sum_i a_ib_i\leq (\max_i a_i) \cdot \sum_i b_i$ for $b_i\geq 0$,
	and the last step uses the fact that
	\begin{align*}
	\max_{i\in [n]}\left(2^{j}\log(|y_i|+1)\right)\leq 2^{j}\cdot \log m\cdot \min_{i\in [n]}\log(|y_i|+1)\leq \log m\cdot \sum_{i=1}^n \log(|y_i|+1)
	\end{align*}

	Applying Bernstein's inequality, we conclude the proof.
\end{proof}

\subsection{Sketch $\sqrt{|\cdot|}$-Vector Product}
Our algorithm for sketching $\sqrt{|\cdot|}$-Vector product is based on the algorithm established in \cite{bvwy17}.
The algorithm is formally presented in Algorithm~\ref{alg:sqrt-sum}.
We first present an algorithm that approximates the inner product for only non-negative $x$. 
In the theorem, we will show the inner product for general $x$ can be approximated as well.
The high level idea is similar to the $p$-stable distribution algorithm established in \cite{i00}. 
However this algorithm is much simpler in terms of hashing function chosen and distribution design. 
In this algorithm, we used the distribution called $p$-inverse distribution (\cite{bvwy17}) over positive integers such that $\Pr[X\le z] = 1-1/z^p$, where $X$ is the $p$-inverse random variable. 
Then we scale each coordinate of $|x|^{1/p}y$ by a random variable drawn from the $p$-inverse distribution.
After this, we run a count-sketch to find the largest few coordinates in the updating scaled vector.
It can be shown that the median value of these output coordinates serve as a good estimation for the $p$-norm of the vector $|x|^{1/p}y$.
A similar idea of this kind can be found in \cite{andoni2017high}.
For the $\sqrt{\cdot}$-case, we simply chose $p=1/2$ then $\|y\|_{p}^{p}$ is a good estimation to $\sum_{i}|x_i|\sqrt{|y_i|}$.
\begin{theorem}
	\label{thm:sqrt-sketch}
	Given a fixed vector $x\in \R^n$ and number $p\in(0,2]$.
	There exists an one-pass streaming algorithm that makes a single pass over the stream updates  to an underlying vector $y\in \R^n$, and outputs a number $Z$, such that, with probability at least $1-\delta$,
	\begin{align*}
	|Z - \langle x, y^p\rangle| \le \epsilon \sum_{i=1}^n{|x_i||y_i|^p}.
	\end{align*}
	The algorithm uses space $O(\epsilon^{-2} \poly(\log ( n/\delta)))$ (excluding the space of $x$).
\end{theorem}
\begin{proof}
The proof of the this theorem is a straightforward application of the results in \cite{bvwy17} by splitting $x$ into positive and negative parts.
\end{proof}
\begin{algorithm}[!t]
	\caption{	\textsc{PolySum}$(x,p,\epsilon)$ \label{alg:sqrt-sum}}
	\begin{algorithmic}[1]
		\Procedure{Initialize}{$x, p$}\Comment{$x\ge 0$, $z\in \Z_{\ge 0}, 0< p\le 2$}
		\State Let $p$-inverse distribution be defined as $\Pr[ z < x ] = 1 - \frac{1}{x^p}$
		\State Let ${\cal D}$ denote the pairwise independent $p$-inverse distribution.
		\State Let $k\gets \Theta(\epsilon^{-2})$
		\State Implicitly store $n \times k$ matrix $Z$, where $Z_{i,j} \sim {\cal D}$ \Comment{Only needs $\Theta(\log n)$ bits}
		\State Initialize \textsc{CS} as a count-sketch instance with space $\Theta(\epsilon^{-2}\poly\log n)$
		\EndProcedure
		\Procedure{\textsc{Update}}{$i,\Delta$} \Comment{$i\in [n], \Delta \in \R$}
		\For{$j=1 \to k$}
		\State $\textsc{CS}.\textsc{Update}( (i,j), |x_i|^{1/p}Z_{i,j}\cdot \Delta  )$
		\EndFor
		\EndProcedure
		\Procedure{\textsc{Query}}{}
		\State $\wt{y} \gets$\textsc{CS}.\textsc{Query}$()$
		\State $z \leftarrow (\frac{k}{2})$-th largest coordinates of $|\wt(y)|$
		\State \Return $z / 2^{1/p}$
		\EndProcedure
	\end{algorithmic}
\end{algorithm}

\subsection{More General Functions $f$}
\label{sec:more general function}
Furthermore, our framework can be applied to a more general set of functions. This set of function includes nearly all ``nice'' functions for $n$ variables. 
For the ease of representation, we neglect the formal definition of the this set. It can be understood that a function in this set satisfies three properties: slow-jumping, slow-dropping and predictable. Readers that are interested, please refer to \cite{bcwy16}.
Here we give three examples for the the functions that we are able to approximate. For example, $x^2 \cdot 2^{\sqrt{\log x}}, (2+\sin x) x^2, 1/\log(1+x)$. Using our proposed general framework and \cite{bcwy16}, we have the following result,
\begin{theorem}
	Given a vector $x\in \{-1, 0, 1\}^n$, and a function $f$ that satisfies the above regularity condition, then
	there exists a one-pass streaming algorithm that makes a single pass over the stream updates to an underlying vector $y\in \R^n$, and outputs a number $Z$, such that, with probability at least $1-\delta$,
	\begin{align*}
	|Z - \langle x, f(|y|)\rangle| \le \epsilon \sum_{i=1}^n{f(|y_i|)}.
	\end{align*}
	The algorithm uses space $O(\poly(\epsilon^{-1}\log ( n/\delta)))$ (excluding the space of $x$).
\end{theorem}
\begin{proof}
The proof is a straightforward application of \cite{bcwy16} by considering the positive part and negative part of $x$ separately.
\end{proof}
\begin{remark}
	We remark that the algorithm in \cite{bcwy16} is quite complicated but has the potential to be simplified.  
	We also note that $x$ is not necessarily restricted on $\{-1, 0, -1\}$, but the complexity depends on ratio of the absolute values of the maximum non-zero entry and minimum non-zero entry (in absolute value) of $x$.
\end{remark}

\subsection{From Vector Product Sketch to Matrix Product Sketch}
With the $f$-vector product sketch tools established, we are now ready to present the result for sketching the matrix product, $M=f(A) B$.
Notice that each entry $M_{i,j} := \langle f(A_{i}), B_j\rangle$ is an inner product.
Thus our algorithm for the matrix sketch is simply maintaining a $f$-vector product sketch for each $M_{i,j}$.
In our algorithm, we assume that matrix $B$ is given, i.e., hardwired in the algorithm.
Thus, if $B\in\R^{n\times k}$ for some $k\ll n$, we only need to keep up to $\wt{O}(nk)$ inner product sketches, which cost in total $\wt{O}(nk)$ words of space.
For the ease of representation, we present our guarantee for matrix product for $f(x):=\log^c(|x|)$ for some $c$ or for $f(x) = x^{p}$ for $0\le p\le 2$, and for matrix $B\in \{-1, 0, 1\}^{n\times k}$. 
Our results can be generalized to a more general set of functions  and matrix $B$ using the results presented in Section~\ref{sec:more general function}. 
\begin{theorem}
	Given a matrix $B\in \{-1, 0, 1\}^{n\times k}$, and a function $f(x):=\log^c(|x|)$ for some $c$ or $f(x):=|x|^p$ for some $0\le p\le2$, then
	there exists a one-pass streaming algorithm that makes a single pass over the stream updates to an underlying matrix $A\in \R^n$ with updates of absolute value at least $1$\footnote{This gurantees that if $A_{i,j}\neq 0$, then $|A_{i,j}|\ge 1$} and outputs a matrix $\wh{M}$, such that, with probability at least $1-\delta$, for all $i,j$,
	\begin{align*}
	|\wh{M}_{i,j} - M_{i,j}| \le \epsilon \sum_{k=1}^n{f(|A_{i,k}|)}.
	\end{align*}
	The algorithm uses space $O(\epsilon^{-2}nk\poly(\log ( n/\delta)))$.
\end{theorem}
\begin{proof}
	The proof of this theorem is a straightforward application of Theorem~\ref{thm:logsum} and Theorem~\ref{thm:sqrt-sketch}.
\end{proof}

\section{Application in Low Rank Approximations} \label{sec:proof_sv_1}

\subsection{Leverage score and its application on samping}\label{sec:leverage_score}
Classic approaches of low rank approximation first compute the leverage scores of the matrix $M$, and then sample rows of $M$ based these scores.
\begin{definition}[Leverage scores, \cite{w14,bss12}]
Let $U\in \mathbb{R}^{n\times k}$ have orthonormal columns with $n \geq k$. We will use the notation $p_i = u_i^2 / k$, where $u_i^2 = \| e_i^\top U \|_2^2$ is referred to as the $i$-th leverage score of $U$. 
\end{definition}

\begin{definition}[Leverage score sampling, \cite{w14,bss12}]\label{def:leverage_score_sampling}
Given $A\in\mathbb{R}^{n\times d}$ with rank $k$, let $U\in \mathbb{R}^{n\times k}$ be an orthonormal basis of the column span of $A$, and for each $i$ let $p_i$ be the squared row norm of the $i$-th row of $U$. Let $k\cdot p_i$ denote the $i$-th leverage score of $U$. Let $\beta>0$ be a constant and $q=(q_1, \cdots,q_n)$ denote a distribution such that, for each $i\in [n]$, $q_i \geq \beta p_i$. Let $s$ be a parameter. Construct an $n\times s$ sampling matrix $B$ and an $s\times s$ rescaling matrix $D$ as follows. Initially, $B = 0^{n\times s}$ and $D = 0^{s\times s}$. For the same column index $j$ of $B$ and of $D$, independently, and with replacement, pick a row index $i\in [n]$ with probability $q_i$, and set $B_{i,j}=1$ and $D_{j,j}=1/\sqrt{q_i s}$. We denote this procedure \textsc{Leverage score sampling} according to the matrix $A$.
\end{definition}

However approximating these scores is highly non-trivial, especially in the streaming setting. 
Fortunately, it suffices to compute the so-called \emph{generalized leverage scores}, i.e., the leverage scores of a proxy matrix. 
We describe the resulting algorithm (Algorithm~\ref{alg:main_lowrank}) and the intuition here and provide the complete analysis later.

\begin{definition}[generalized leverage score]\label{def:gen_leverage_1}
Consider two accuracy parameters $\alpha \in (0,1) , \delta \in (0,1)$, and two positive integers $q$ and $k$ with $q \geq k$. 
If there is a matrix $E\in \R^{n\times n}$ with rank $q$ and that approximates the row space of $M\in \R^{n\times n}$ as follows, 
\begin{align*}
\exists X,  \|XE - M\|_F \leq (1+\alpha) \| M - [M]_k \|_F + \delta,
\end{align*}
then the leverage scores of $E$ are called a set of $(1+\alpha, \delta, q, k)$-generalized leverage scores of $A$. 
Suppose $E$ has an SVD decomposition $U\Sigma V^\top$, where $U \in \R^{n\times q}, V \in \R^{n\times q}$  are orthonormal matrices, then its leverage scores are $\ell_i = \|V^i\|_2^2$ where $V^i$ is the $i$-th row of $V$, $\forall i \in [n]$. 
\end{definition}

These scores can be computed easier. We first need to find such an matrix $E$.
Let $S$ be a \emph{subspace embedding matrix} (i.e.,
 \begin{align*}
 \|SM x\|_2\in (1\pm\alpha) \|Mx\|_2, ~~~ \forall x\in \R^{n},
 \end{align*}
a sufficient large matrix with random $+1,-1$ entries will have this property). 

Then $E=SM$ satisfies the requirement in Definition~\ref{def:gen_leverage_1}, and thus we can simply use our sketching method to approximate $SM$ and then compute its leverage scores. In Algorithm~\ref{alg:main_lowrank}, we will use the concatenation of the positive and negative parts of $S$, 
since it also satisfies the requirement and empirically has better accuracy than $S$.
The quality of the generalized scores (i.e., $\alpha$ and $\delta$) will depend on the parameter $s$ in the algorithm that are specified in our final Theorem~\ref{thm:main_lowrank}. 

The scores then can be used for sampling.
Let $P$ be a set of columns of $M$ sampled based on these scores (defined in Line~\ref{alg:sample p} of Algorithm~\ref{alg:main_lowrank}). 
It is known that, when the scores are $ ( O(1), 0, q, k)$-generalized leverage scores, then the span of a $P$ with $\Omega(q\log q )$ columns will contain a rank-$q$ matrix which provides a $O(1)$-approximation to $M$~\cite{DriMagMahWoo12,bss12,bw14,cemmp15,swz19}. It is tempting to set $q = k$ to match our final goal of rank-$k$ approximation, but all existing fast methods require $q > k$.
To improve the rank-$q$ to rank-$k$, we use \emph{adaptive sampling}.

Adaptive sampling samples some extra columns from $M$ according to their squared distances to the span of $P$. For a column $M_{*i}$, we thus need to use our sketching method to estimate $\|M_{*i}\|_2^2 - \|\Gamma_{*i}\|_2^2$, where $\Gamma_{*i}$ is its projection on to the span of $P$. This introduces some additive errors but they can be handled by thresholding. 
Let $\tilde Y$ be the sampled columns.
Adaptive sampling ensures that there is a good rank-$k$ approximation in the span of $Y:=\tilde Y \cup P$ as long as we have sampled sufficiently many columns. 
To obtain our final rank-$k$ approximation, it suffices to project $M$ to the span of $Y$ and compute the top $k$ singular vectors. 
The projection can be done by sketching and the errors are, again, small.

\subsection{Proof of Theorem \ref{thm:main_lowrank}}
Recall that there are three steps in computing the top singular vectors (see Algorithm~\ref{alg:main_lowrank}): 
\begin{itemize}
	\item Compute the generalized leverage scores and sample a set $P$ according to the scores,
	\item Adaptive sampling to get a set $Y$,
	\item Project to the span of $Y$ and compute the approximation solution there. 
\end{itemize}
Below we present the complete proofs for each step.

For simplicity, we use the following notion.
\begin{definition}
We say that the span of $P$ has a $(1+ \epsilon, \Delta)$-approximation subspace for $M$ if there exists $C$ such that
\begin{align*}
  \| PC - M\|_F \leq (1+\epsilon) \|M - [M]_k\|_F + \Delta.
\end{align*}
\end{definition}

\subsection{Sampling by Generalized Leverage Scores}

First, recall the definition of generalized leverage scores and related property from~\cite{balcan2016communication}.


\begin{lemma}[Lemma 2 in~\cite{balcan2016communication}] \label{lem:leverage_sample}
Suppose $0< k\leq q \le m \le n$, $\alpha>0$, $\Delta>0$, and $A \in \R^{n \times n}$. 
Let $B \in \R^{n \times b}$ be $b = O( \alpha^{-2} q \log q )$ columns sampled from $A$ according to a set of $(1+\alpha, \Delta, q, k)$-generalized leverage scores of $A$.
Then with probability $\ge 0.99$, the col-span of $B \in \R^{n \times b}$ has a rank-$q$ $(1+ 2\alpha, 2\Delta)$-approximation subspace for $A$. That is, there exists $C \in \R^{b \times n}$ such that
\begin{align*}
  \| BC - A\|_F \leq (1+2\alpha) \|A - [A]_k\|_F + 2\Delta.
\end{align*}{}
\end{lemma}

We need the following result about subspace embedding.
\begin{lemma}[Lemma 3 in \cite{balcan2016communication}]\label{lem:lemma3b16}
We say $S$ is a $(1+\epsilon,\Delta)$-good subspace embedding if it satisfies the following. \\
(Subspace Embedding). For any orthonormal $V \in \R^{n \times k}$ (i.e. $V^\top V = I$ ), 
\begin{align*}
(1-c) \| V x\| \leq \| S V x \| \leq (1+c) \| V x \|
\end{align*}
where $c \in (0,1)$ is a sufficiently small constant. \\
(Approximate Matrix Product). For any fixed $A \in \R^{n \times n}$ and $B \in \R^{n \times k}$
\begin{align*}
\| A^\top S^\top S B - A^\top B \|_F^2 \leq \frac{\epsilon}{k} \| A \|_F^2 \cdot \| B \|_F^2 + \Delta.
\end{align*}
\end{lemma}

We are going to show that in Algorithm~\ref{alg:main_lowrank}, the span of $P$ has a good approximation subspace.  Intuitively, $E = R \cdot M \in \R^{2 s \times n}$ approximates the row space of $M \in \R^{n \times n}$ and $\tilde E \in \R^{2s \times n}$ approximates $E \in \R^{2s \times n}$, so by the definition, the leverage scores $\{\ell_i\}$ of $\tilde E \in \R^{2s \times n}$ are the generalized leverage scores of $M$. Then the conclusion follows from Lemma~\ref{lem:leverage_sample}. Formally, we have the following lemma.

\begin{lemma}[Sampling leverage scores]\label{lem:sample_score}
Let $s = O(k\log k)$. Let $d_1 = O( k \log^2 k )$. 
Recall that $P\in \R^{n\times d_1}$ is the matrix sampled with the leverage score of $(S \cdot M) \in \R^{s \times n}$, as constructed in Line \ref{alg:sample p} as in Algorithm \ref{alg:main_lowrank}.

There exists matrix $S\in \mathbb{R}^{s\times n}$, such that there exists $C$ satisfying
\begin{align*}
  \| PC - M\|_F^2 \leq 5 \|M - [M]_k\|_F^2 + \Delta_1.
\end{align*}
where $\Delta_1 = O(\epsilon^2 /s) \|M\|_{1,2}^2$. 
\end{lemma}

\begin{proof}
First, $s$ is large enough so that $S \in \R^{s \times n}$ is a $0.1$-subspace embedding matrix for subspace of dimension $k$; see~\cite{balcan2016communication,w14}. Then it is known that there exists $Z\in \mathbb{R}^{n\times s}$ satisfying
\begin{align*}
  \| ZSM - M\|_F \leq (1+0.1) \|M - [M]_k\|_F.
\end{align*}

Clearly, there exists $X=[Z, Z] \in \R^{n\times 2s}$ such that 
\begin{align}\label{eq:XRM}
  \| XRM - M\|_F \leq (1+0.1) \|M - [M]_k\|_F.
\end{align}
Let $E=RM \in \R^{2s\times n}$. Then 
\begin{align*}
\| M - X \tilde E\|_F^2 
& \le 2\| M - X E \|_F^2 + 2\| X E - X \tilde E\|_F^2 \\
& \le 3\|M - [M]_k\|_F^2 + 2\| X E - X \tilde E\|_F^2.
\end{align*}
where the first step follows from the inequality $\|A+B\|_F^2\leq (\|A\|_F+\|B\|_F)^2\leq 2\|A\|^2_F+2\|B\|^2_F$, the second step follows from Eq.~\eqref{eq:XRM} and the definition $E=RM \in \R^{2s \times n}$.

Consider the second term.
\begin{align*}
\| X E - X \tilde E\|_F
\le & ~ \| X \|_2 \|E - \tilde E\|_F \\
\le & ~ 2\| Z \|_2 \|E - \tilde E\|_F.
\end{align*}
where the last step we use $X=[Z, Z]\in \mathbb{R}^{n\times 2s}$.

Hence we have 
\begin{align*}
\|E - \tilde E\|_F 
\le & ~ O(\epsilon) \max_{i\in [2s],j\in [n]} |R_{i,j}|\|M\|_{1,2}  \\
\le & ~ O(\epsilon/\sqrt{s})\|M\|_{1,2}.
\end{align*}
where the first step follows from  our guarantee on our sketching method in Theorem~\ref{thm:logsum}, the second step follows from the construction of $R$, i.e. the range of each entry of the CountSketch matrix.

By Lemma \ref{lem:lemma3b16}, we can rewrite $Z$ as follows:
\begin{align*}
Z = & ~ [M]_k (S[M]_k)^\dagger \\
= & ~ [M]_k[M]_k^{\dagger}S^{\dagger},
\end{align*}

We can upper bound $\| Z \|_2$ by $O(1)$,
\begin{align*}
\|Z\|_2 \leq & ~ \|[M]_k[M]_k^{\dagger}\|_2\cdot\|S^{\dagger}\|_2 \\
= & ~ \|S^{\dagger}\|_2 \\
= & ~ O(1).
\end{align*} 
Putting it all together, we have
\begin{align*}
\| M - X \tilde E\|_F^2 
 \le 3\|M - [M]_k\|_F^2 + O(\epsilon^2 /s) \|M\|_{1,2}^2.
\end{align*}
This satisfies the definition of generalized leverage scores. Then the statement follows from Lemma~\ref{lem:leverage_sample}.
\end{proof}

\subsection{Adaptive Sampling}

\begin{lemma}[Adaptive]\label{lem:adaptive}
Let  $d_1=O(k\log^2 k)$ and $ d_2 = O(k/\epsilon)$.
With probability $\ge 0.99$, 
there exists $C\in \R^{(d_1+d_2)\times n}$ such that $YC$ is rank-$k$ and 
\begin{align*}
 \|YC - M\|_F^2  \le 5\| M - [M]_k\|_F^2 + \Delta_1 + \Delta_2
\end{align*}
where $\Delta_1$ is defined as Lemma~\ref{lem:sample_score}, and $\Delta_2 = O(\epsilon \sqrt{d_1} +  \epsilon^2 d_1)\|M\|^2_F$.
\end{lemma}

\begin{proof}
If $p_i$'s are larger than a constant times the true square distances $s_i$'s, then the statement follows from Theorem~\ref{thm:theorem_3_in_dv06}. So consider the difference between $\tilde s_i$ and $s_i$. 

Let $\Gamma = Q_p^\top M\in \R^{d_1\times n}$ where $Q_p$ is obtained from QR-decomposition as in Line \ref{alg:qrdecomposition} of Algorithm \ref{alg:main_lowrank}. 
\begin{align*}
|\tilde s_i - s_i| 
& = \Big| \| \Gamma_i \|_2^2 - \| \tilde\Gamma_i \|_2^2 + \| \tilde M_i \|_2^2 - \| M_i \|_2^2 \Big|.
\end{align*}
By our guarantee in Theorem~\ref{thm:logsum}, 
\begin{align}\label{eq:gammaji}
|\Gamma_{ji}  - \tilde\Gamma_{ji}  |
& \le \epsilon \| (Q_p)_j\|_F \|M_i\|_F = \epsilon  \|M_i\|_2
\end{align}
where the last inequality follows since $(Q_p)_j$'s are basis vectors and have length $1$. 

So 
\begin{align}\label{eq:gammaji_square}
\left|\Gamma_{ji}^2  - \tilde\Gamma_{ji}^2 \right|
= & ~ \left| ( \Gamma_{ji} - \tilde\Gamma_{ji} ) \cdot ( \Gamma_{ji}  + \tilde\Gamma_{ji} ) \right| \notag\\
= & ~ \left|\Gamma_{ji} - \tilde\Gamma_{ji} \right| \cdot \left|\Gamma_{ji}  + \tilde\Gamma_{ji} \right| \notag\\
\leq & ~ \left|\Gamma_{ji} - \tilde\Gamma_{ji} \right| \cdot \left( \left|\Gamma_{ji}  - \tilde\Gamma_{ji} \right|+2|\Gamma_{ji}| \right) \notag\\
= & ~  2  \left|\Gamma_{ji} - \tilde\Gamma_{ji} \right| \cdot |\Gamma_{ji}| + \left|\Gamma_{ji} - \tilde\Gamma_{ji} \right|^2 \notag \\
\le & ~ 2\epsilon \|M_i\|_2 |\Gamma_{ji}| + \epsilon^2  \|M_i\|_2^2
\end{align}
where the third step follows from triangle inequality, and the last step follows Eq.~\eqref{eq:gammaji}.

And
\begin{align*}
\Big| \|\Gamma_{i}\|_2^2  - \|\tilde\Gamma_{i}\|_2^2  \Big|
&\leq \sum_{j\in [d_1]} |\Gamma_{ji}^2  - \tilde\Gamma_{ji}^2|\\
& \le 2\epsilon  \|M_i\|_2 \sum_{j\in [d_1]} |\Gamma_{ji}| + O(\epsilon^2 d_1) \|M_i\|_2^2  \\
& \le O(\epsilon  \sqrt{d_1} ) \|M_i\|_2 \|\Gamma_{i}\|_2 + O(\epsilon^2 d_1) \|M_i\|_2^2  \\
& \le O(\epsilon  \sqrt{d_1}) \|M_i\|^2_2 + O(\epsilon^2 d_1) \|M_i\|_2^2 \\
& = O(\epsilon  \sqrt{d_1} + \epsilon^2 d_1) \|M_i\|_2^2. 
\end{align*}
where the first step follows from triangle inequality, the second step follows from \eqref{eq:gammaji_square}, the third step follows form Cauchy-Swartz inequality, the fourth step follows $\Gamma_{i}=Q_p^\top M_i\in \R^{d_1}$ and $Q_p\in \R^{d_1\times n}$ is an orthonormal matrix. 

Therefore,
\begin{align*}
|\tilde s_i - s_i| \le O(\epsilon  \sqrt{d_1} + \epsilon^2 d_1) \|M_i\|^2_2 := \delta_i.
\end{align*}

Suppose that the algorithm sets $\eta:=O(\epsilon \sqrt{d_1} +  \epsilon^2 d_1)$ such that the threshold $\delta_i\le \eta \tilde z_i \le 2\delta_i$. 

If $\tilde s_i \le \delta_i$, then $s_i \le 2 \delta_i$, and thus $p_i \ge s_i/2$. If $\tilde s_i > \delta_i$, then $s_i \le \tilde s_i + \delta_i \le 2 \tilde s_i \le 2 p_i$, and thus $p_i \ge s_i/2$. 
Now, if $\sum_i p_i$ is not too large compared to $\sum_i s_i$, then we are done by applying Theorem~\ref{thm:theorem_3_in_dv06}. 

Let $C\subseteq [n]$ denote the set of indices $i$ such that $\tilde s_i \ge \delta_i$. 
If $\sum_{i \in C} \tilde s_i \ge 2 \sum_{i \in [n]} \delta_i $, then 
\begin{align}\label{eq:s_vs_tilde_s}
\sum_{i \in C} s_i \ge \sum_{i \in C} \tilde s_i - \sum_{i \in [n]} \delta_i \ge \frac{1}{2}\sum_{i \in C} \tilde s_i 
\end{align}
and thus 
\begin{align*}
  \sum_{i \in [n]} p_i 
  \le & ~ \sum_{i \in C} \tilde s_i + 2 \sum_{i \in [n]} \delta_i \notag \\
  \le & ~ 2 \sum_{i \in C} \tilde s_i \notag \\
  \le & ~ 4 \sum_{i \in C} s_i \notag \\
  \le & ~ 4\sum_{i \in [n]} s_i.
\end{align*}
where the first step follows from the definition of $p_i$,
i.e. $p_i=\max\{ \tilde s_i,\eta\tilde z_i\}$ and  the assumption $\eta\tilde z_i \le 2\delta_i$, the second step follows from the assumption $\sum_{i \in C} \tilde s_i \ge 2 \sum_{i \in [n]} \delta_i $, the third step uses \eqref{eq:s_vs_tilde_s}, the fourth step  is because $s_i\geq 0$ and $C\subset [n]$.

So we are done in this case.

In the other case when $\sum_{i \in C} \tilde s_i < 2 \sum_{i \in [n]} \delta_i $, we have 
\begin{align*}
  \sum_{i \in [n]}  s_i
  \leq & ~ \sum_{i \in [n]} \tilde s_i+\delta_i \\
  \leq & ~ \sum_{i\in C} 2s_i+\sum_{i\not\in C}2\delta_i \\
  \leq & ~ 4 \sum_{i \in [n]} \delta_i  \\
  = & ~ O(\epsilon \sqrt{d_1} +  \epsilon^2 d_1)\|M\|^2_F :=\Delta_2.
\end{align*}
where the first step follows from $\delta_i=|s_i-\tilde s_i|$ and triangle inequality, the second step uses the construction of the set $C$ and the third step uses the assumption $\sum_{i \in C} \tilde s_i < 2 \sum_{i \in [n]} \delta_i$.

This means that $\Gamma$ is close to $M$, and thus $[\Gamma]_k$ (the best rank-$k$ approximation to $\Gamma$) will be the desired approximation in the span of $P$ (and thus the span of $Y$ since $P\subseteq Y$). This completes the proof.
\end{proof}

\subsection{Computing Approximation Solutions}

\begin{lemma} \label{lem:approxLA}

Let $d_1=O(k\log^2k)$  and $d_2=O(k/\epsilon)$.
There is an algorithm that outputs a matrix $L$ , such that 
\[
	\|LL^\top M - M\|_F^2  \leq 10\| M - [M]_k \|_F^2 + 2\Delta_1 + 2\Delta_2 + \Delta_3
\]
where $\Delta_1$ is defined as Lemma~\ref{lem:sample_score}, $\Delta_2$ is defined as Lemma~\ref{lem:adaptive} and $\Delta_3 = O(\epsilon^2 (d_1+d_2))\| M\|_F^2$.
\end{lemma}

\begin{proof}
Since $Q\in \mathbb{R}^{n\times (d_1+d_2)}$ is orthonormal,
$Q^\top Q=I_{d_1+d_2}$.
We need the following auxiliary result: for any $A\in \R^{(d_1+d_2)\times n}$,
\begin{align}\label{eq:error} 
  \| QAM - M \|_F^2  = \| QA  M - QQ^\top M \|_F^2  +  \| M - Q Q^\top M \|_F^2. 
\end{align}
This is because
\begin{align*}
& ~ \left(QAM - QQ^\top M\right)^\top \left(M - Q Q^\top M\right)\\
= & ~ M^\top A^\top Q^\top M-M^\top QQ^\top M-M^\top A^\top Q^\top Q Q^\top M+MQQ^\top QQ^\top M\\
= & ~ M^\top A^\top Q^\top M-M^\top QQ^\top M-M^\top A^\top  Q^\top M+MQQ^\top M\\
= & ~ 0
\end{align*}
where the second step uses  $Q^\top Q=I_{d_1+d_2}$.

Then \eqref{eq:error} simply follows from Theorem \ref{thm:pythagorean}.

We also need the following result: for any $A\in \R^{(d_1+d_2)\times n}$,
\begin{align}\label{eq:Q_keep_norm}
\|QA\|_F^2=\|A\|_F^2
\end{align}
This is because
\begin{align*}
\|QA\|_F^2=\Tr[A^\top Q^\top Q A]=\Tr[A^\top A]=\|A\|_F^2
\end{align*}
where the second step we uses the fact that $Q^\top Q=I_{d_1+d_2}$.

Let $X\in \R^{n\times n}$ denote the $YC\in \R^{n\times n}$ in Lemma~\ref{lem:adaptive}.
Recall that $Q_y$ is obtained from QR-decomposition of $Y\in \R^{n\times (d_1+d_2)}$,
so we can write $Y=Q_yR_y$.
For simplicity, let $Q$ denote $Q_y$ and $R$ denote $R_y$. 
Then we have
\begin{align}\label{eq:QQMk}
\| Q[Q^\top M]_k   - M \|^2_F 
 = & \| Q[Q^\top M]_k - QQ^\top M \|^2_F  + \| QQ^\top M - M \|_F^2 \notag\\
 =&\|[Q^\top M]_k - Q^\top M \|^2_F  + \| QQ^\top M - M \|_F^2 \notag\\
\le & \|RC - Q^\top M \|^2_F  + \| QQ^\top M - M \|_F^2 \notag\\
=& \|QRC - QQ^\top M \|^2_F  + \| QQ^\top M - M \|_F^2 \notag\\
=& \|QRC - M \|^2_F \notag\\
= & \|X - M \|^2_F
\end{align}
where the first step uses \eqref{eq:error} by setting $A=[Q^\top M]_k$, 
the second step uses \eqref{eq:Q_keep_norm},
the third step uses the fact that $\rank(RC)\leq \rank(QRC)=\rank(Y)\leq k$ and $[Q^\top M]_k\in \R^{(d_1+d_2)\times n}$ is the best rank-$k$ approximation for $Q^\top M\in \R^{(d_1+d_2)\times n}$, 
the fourth step again uses Eq.~\eqref{eq:Q_keep_norm},
the fifth step uses the Eq.~\eqref{eq:error} by setting $A=RC$,
and the last step uses that $QRC=YC=X \in \R^{n \times n}$.

Therefore,
\begin{align} \label{eq:inQ}
\|Q[Q^\top M]_k  - M\|_F^2  
\leq & ~ \|M - X\|^2_F  \notag\\
\leq & ~ 5\|M - [M]_k\|^2_F + \Delta_1+\Delta_2. 
\end{align}
where the first step follows from \eqref{eq:QQMk}, the second step follows from Lemma \ref{lem:adaptive}.

Let $W\in\R^{(d_1+d_2)\times k}$ denote the top $k$ singular vectors of $Q^\top M\in\R^{(d_1+d_2)\times n}$. Since $\tilde W\in\R^{(d_1+d_2)\times k}$ are the top $k$ singular vectors of $\tilde\Pi\in \R^{(d_1+d_2)\times n}$, we have
\begin{align}\label{eq:inW}
 & ~ \|(\tilde{W}\tilde{W}^\top - I) Q^\top M\|_F^2 \\
 \leq& ~ 2\|(\tilde{W}\tilde{W}^\top - I) \tilde\Pi\|_F^2 +  2\|(\tilde{W}\tilde{W}^\top - I) (Q^\top M - \tilde\Pi)\|_F^2 \notag\\
\leq & ~ 2\|(WW^\top - I) \tilde \Pi\|_F^2 +  2\|(\tilde{W}\tilde{W}^\top - I) (Q^\top M - \tilde\Pi)\|_F^2 \notag\\
\le & ~  2\|(WW^\top - I) Q^\top M\|_F^2 + 2\|(WW^\top - I) (Q^\top M - \tilde \Pi)\|_F^2 +  2\|(\tilde{W}\tilde{W}^\top - I) (Q^\top M - \tilde\Pi)\|_F^2\notag  \\
\le & ~ 2\|[Q^\top M]_k - Q^\top M\|_F^2 + 2\|WW^\top - I\|_2^2\| (Q^\top M - \tilde \Pi)\|_F^2  + 2\|\tilde {W}\tilde{W}^\top - I\|_2^2\| (Q^\top M - \tilde \Pi)\|_F^2\notag \\
\le & ~ 2\|[Q^\top M]_k - Q^\top M\|_F^2 +  4\| (Q^\top M - \tilde \Pi)\|_F^2 \notag \\
\le & ~ 2\|[Q^\top M]_k - Q^\top M\|_F^2 +  O(\epsilon^2)\|Q\|_F^2\|M\|_F^2 \notag \\
\le & ~ 2\|[Q^\top M]_k - Q^\top M\|_F^2 +  O(\epsilon^2 ) \cdot ( d_1+d_2)\|M\|_F^2 \notag \\
=& ~ 2\|[Q^\top M]_k - Q^\top M\|_F^2 +  O(\epsilon^2 (d_1+d_2))\|M\|_F^2 
\end{align}
where the first step uses the fact that $\|A+B\|_F^2\leq 2\|A\|_F^2+2\|B\|_F^2$,
the second step uses the fact that $(\tilde{W}\tilde{W}^\top - I) Q^\top M=[Q^\top M]_k-Q^\top M$ and $[Q^\top M]_k$ is the best rank $k$ approximation of $Q^\top M$,
the third step uses the fact that $\|AB\|_F\leq \|A\|_2\cdot \|B\|_F$ for any matrices $A,B$,
and $WW^\top Q^\top M=[Q^\top M]_k$, since $W$ are the top $k$ singular vectors of $Q^\top M$,
the fourth step uses $\| AA^\top - I \|_2\leq 1$ for all orthonormal matrix $A\in \R^{(d_1+d_2)\times k}$ since $(AA^\top - I)^2=I-AA^\top$,
the fifth step uses convergence grantee in Theorem \ref{thm:logsum},
the sixth step follows from $Q$ is an orthonormal matrix with $d_1+d_2$ columns.

We now bound the error using the above two claims.

Noting $L=Q\tilde W\in \R^{(d_1+d_2)\times k}$, hence  by \eqref{eq:Q_keep_norm} we have
\begin{eqnarray}\label{eq:LLM}
  \|LL^\top  M - QQ^\top M\|_F^2 = \|\tilde{W}\tilde{W}^\top Q^\top M - Q^\top M\|_F^2 
\end{eqnarray}
Therefore
\begin{eqnarray*}
\|LL^\top M - M\|_F^2  
& = & \|LL^\top  M - QQ^\top M\|_F^2  +  \|M - Q Q^\top M\|_F^2\\
& \leq &  2\|Q[Q^\top M]_k  - QQ^\top M\|_F^2 +  O(\epsilon^2 (d_1+d_2))\|M\|_F^2
  +  \|M - Q Q^\top M\|_F^2  \\
& \leq &  2\|Q[Q^\top M]_k   - M \|_F^2 +   O(\epsilon^2 (d_1+d_2))\|M\|_F^2\\
&\leq & 10 \|M - [M]_k\|^2_F + 2\Delta_1 + 2\Delta_2 +  O(\epsilon^2 (d_1+d_2))\|M\|_F^2\\
&= &10 \|M - [M]_k\|^2_F + 2\Delta_1 + 2\Delta_2 +\Delta_3,
\end{eqnarray*}

where the first step uses the fact that $L=Q\tilde W$ and \eqref{eq:error} with $A=\tilde W L^\top$,
the second step uses \eqref{eq:LLM} and \eqref{eq:inW},
the third step uses \eqref{eq:error} with $A=[Q^\top M]_k$,
the fourth step uses Lemma \ref{lem:adaptive},
and the last step is the definition of $\Delta_3$.
\end{proof}

\subsection{Main result}

\begin{table}\caption{Table of parameters}\label{tab:para}
\centering
\begin{tabular}{ |l|l|l|l| } 
\hline
{\bf Notation} & {\bf Choice} & {\bf Location} & {\bf Comment} \\\hline
$s$ & $O(k\log k)$ & Lemma~\ref{lem:sample_score} & size of oblivious sketching matrix\\ \hline

$d_1$ & $O(k \log^2 k)$ & Lemma~\ref{lem:sample_score} & size of column sampling matrix \\ \hline

$d_2$ & $O(k/\epsilon)$ & Lemma~\ref{lem:adaptive} & size of adaptive column sampling \\ \hline
$\Delta_1$ & $O(\epsilon^2/s) \| M \|_{1,2}^2$ & Lemma~\ref{lem:sample_score} & error from oblivious sketching matrix \\ \hline
$\Delta_2$ & $O(\sqrt{\epsilon^2d_1} + \epsilon^2 d_1) \| M \|_F^2$ & Lemma~\ref{lem:adaptive} & error from column sampling matrix  \\ \hline
$\Delta_3$ & $O(\epsilon^2 (d_1+d_2)) \|M\|_F^2$ & Lemma~\ref{lem:approxLA} & error from adaptive column sampling\\ \hline

\end{tabular}
\end{table}

\begin{theorem}
There exists an algorithm (procedure \textsc{LowRankApprox} in Algorithm~\ref{alg:main_lowrank}) that with parameter settings as in Table \ref{tab:para},
runs in query time $nk\poly(\log n,1/\epsilon)$ and space $\tilde O(nk/\epsilon^2)$, 
outputs a matrix $L\in \R^{n\times k}$ so that
\[
	\|LL^\top M - M\|_F^2  \leq 10\| M - [M]_k \|_F^2 + O(\epsilon d_1)\|M\|_F^2 + O(\epsilon^2 /s) \|M\|_{1,2}^2.
\]
holds with probability at least $9/10$. 
\end{theorem}

\begin{proof}[Proof of guarantee]
\begin{align*}
	&~\|LL^\top M - M\|_F^2  \\
	\leq& ~ 10\| M - [M]_k \|_F^2 + 2\Delta_1 + 2\Delta_2 + \Delta_3\\
	= &~10\| M - [M]_k \|_F^2 + O(\epsilon^2/s)\|M\|_{1,2}^2 + O(\sqrt{\epsilon^2d_1} + \epsilon^2 d_1) \| M \|_F^2 + O(\epsilon^2 (d_1+d_2))\|M\|_F^2\\
	= &~10\| M - [M]_k \|_F^2 + O(\epsilon^2/s)\|M\|_{1,2}^2 + O(\sqrt{\epsilon^2d_1} + \epsilon^2 d_1) \| M \|_F^2 + O(\epsilon d_1)\|M\|_F^2\\
	= &~10\| M - [M]_k \|_F^2 + O(\epsilon^2/s)\|M\|_{1,2}^2 + O(\epsilon d_1)\|M\|_F^2
\end{align*}
where the first step uses Lemma \ref{lem:approxLA},
the third step uses the definition of $d_1=O(k\log^2 k)$ and $d_2=O(k/\epsilon)$ so $d_1+d_2=O(d_1/\epsilon)$.
\end{proof}
\begin{proof}[Proof of time and space]
The largest matrix we ever need to store during the process has size $n\times (d_1+d_2)=O(\epsilon^{-1}nk\log^2 k)$.
The space needed by $\textsc{LogSum}$ is bounded by $\tilde O(\epsilon^{-2} nk)$ by Theorem \ref{thm:logsum}.
So the overall space used is at most $\tilde O(\epsilon^{-2} nk)$.

Since we only call $\textsc{LogSum}$ 4 times in the whole process, the query time hence follows from Theorem \ref{thm:logsum}.
\end{proof}

Notice that $\|M\|_{1,2}\geq \|M\|_F \geq \| M - [M]_k \|_F$,
so we can rescale $\epsilon$ to get Theorem \ref{thm:main_lowrank}.

\section{Examples Demonstrating the Differences Between $A$ and $\log(A)$} \label{sec:diff}

\subsection{$\rank (A) \gg \rank ( \log A)$}
In this section, we provide a matrix $A \in \R^{n \times n}$ with $\rank$-$n$, however, the $\rank(\log A) = 1$.

Recall the definition of Vandermonde matrix.
\begin{definition}
An $m \times n$ Vandermonde matrix usually is defined as follows
\begin{align*}
V = \begin{bmatrix}
1 & \alpha_1 & \alpha_1^2 & \cdots & \alpha_1^{n-1} \\
1 & \alpha_2 & \alpha_2^2 & \cdots & \alpha_2^{n-1} \\
1 & \alpha_3 & \alpha_3^2 & \cdots & \alpha_3^{n-1} \\
\vdots & \vdots & \vdots & \ddots & \vdots \\
1 & \alpha_m & \alpha_m^2 & \cdots & \alpha_m^{n-1} \\
\end{bmatrix}
\end{align*}
or $V_{i,j} = \alpha_i^{j-1}, \forall i \in [m], j \in [n]$
\end{definition}

\begin{theorem}\label{thm:rankA_is_big_ranklogA_is_small}
Let $A$ denote a $n \times n$ Vandermonde matrix with $\alpha_i \neq \alpha_j, \forall i\neq j$. Then $\rank(A) = n$ and $\rank ( \log (A) ) = 1$.
\end{theorem}
\begin{proof}
By definition of $A$, we have,
\begin{align*}
\begin{bmatrix}
1 & \alpha_1 & \alpha_1^2 & \cdots & \alpha_1^{n-1} \\
1 & \alpha_2 & \alpha_2^2 & \cdots & \alpha_2^{n-1} \\
1 & \alpha_3 & \alpha_3^2 & \cdots & \alpha_3^{n-1} \\
\vdots & \vdots & \vdots & \ddots & \vdots \\
1 & \alpha_n & \alpha_n^2 & \cdots & \alpha_n^{n-1} \\
\end{bmatrix}
\end{align*}
Note that, we can compute the determinant of matrix $A$,
\begin{align*}
\det(A) = \prod_{1 \leq i < j \leq n} (\alpha_j - \alpha_i).
\end{align*}
Since $\alpha_j \neq \alpha_i, \forall j \neq i$, thus $\det(A) \neq 0$ which implies $\rank(A) = n$.

By definition of $\log (A)$, we have,
\begin{align*}
\log (A) = & ~ 
\begin{bmatrix}
0 & \log(\alpha_1) & 2\log(\alpha_1) & \cdots & (n-1)\log(\alpha_1) \\
0 & \log(\alpha_2) & 2\log(\alpha_2) & \cdots & (n-1)\log(\alpha_2) \\
0 & \log(\alpha_3) & 2\log(\alpha_3) & \cdots & (n-1)\log(\alpha_3) \\
\vdots & \vdots & \vdots & \ddots & \vdots \\
0 & \log(\alpha_n) & 2\log(\alpha_n) & \cdots & (n-1)\log(\alpha_n) \\
\end{bmatrix} \\
= & ~ 
\begin{bmatrix}
\log(\alpha_1) \\
\log(\alpha_2) \\
\log(\alpha_3) \\
\vdots \\
\log(\alpha_n) \\
\end{bmatrix}
\cdot
\begin{bmatrix}
0 & 1 & 2 & \cdots & (n-1)
\end{bmatrix}.
\end{align*}
Therefore $\rank(\log (A)) = 1$.
\end{proof}

\subsection{$\rank(A) \ll \rank(\log A)$}
In this section, we provide a matrix $A \in \R^{n\times n}$ with $\rank$-$n/2$, however, the $\rank(\log A) = n$.
\begin{theorem}\label{thm:rankA_is_small_ranklogA_is_big}
There is a matrix $A\in \R^{n\times n}$ such that $\rank(A) = n/2$ and $\rank(\log(A)) = n$.
\end{theorem}

\begin{proof}
Let $B$ denote a $2 \times 2$ matrix as follows
\begin{align*}
B = \begin{bmatrix}
1 & 2 \\
2 & 4
\end{bmatrix}
\end{align*}
It is not hard to see that $\rank(B) = 2$ and $\rank(\log (B))=1$. We define matrix $A$ by copying $B$ by $n/2$ times on $A$'s diagonal blocks, 
\begin{align*}
A = \begin{bmatrix}
B & 0 & 0 & \cdots & 0 \\
0 & B & 0 & \cdots & 0 \\
0 & 0 & B & \cdots & 0 \\
\vdots & \vdots & \vdots & \ddots & \vdots \\
0 & 0 & 0 & \cdots & B
\end{bmatrix}.
\end{align*}
Then we have $\rank(A) = n/2$ and $\rank(\log (A)) = n$.
\end{proof}

Due to the following fact, copying $\rank$-$1$ matrix several times won't give a better theorem~\ref{thm:rankA_is_small_ranklogA_is_big}.
\begin{fact}
For any $\rank$-$1$ matrix $A$, the $\rank(\log (A)) \leq 2$.
\end{fact}
\begin{proof}
Without loss of generality, let's assume $A$ can be written as
\begin{align*}
A = \alpha^\top \beta = \begin{bmatrix}
\alpha_1 \\
\alpha_2 \\
\vdots \\
\alpha_n
\end{bmatrix}
\cdot
\begin{bmatrix}
\beta_1 & \beta_2 & \cdots \beta_n
\end{bmatrix}
\end{align*}
Let $B$ denote $\log(A)$, then it is easy to that $B_{i,j} = \log(\alpha_i) + \log(\beta_i)$. Therefore matrix $B$ can be decomposed into the following case
\begin{align*}
B = \log (\alpha) \cdot {\bf 1} + {\bf 1}^\top \log (\beta)^\top
\end{align*}
Thus, $\rank(B) \leq 2$.
\end{proof}


\section{Application of $f$-Matrix Product Sketch in Linear Regression} \label{sec:more application}
In this section, we consider the application to linear regression. Linear regression is a fundamental problem in machine learning, and there is a long line of work using sketching/hashing idea to speed up the running time \cite{cw13,mm13,psw17,lhw17,alszz18,dssw18,swz19,cww19}.

\begin{algorithm}
\begin{algorithmic}[1]\caption{Linear regression by $f$-matrix product sketch}\label{alg:linear_regression}
\Procedure{\textsc{LinearRegression}}{$M,b,n,d,$} \Comment{Theorem~\ref{thm:linear_regression}}
	\State Implicitly form $A = \log M$ \Comment{$A\in \R^{n \times d}$}
	\State $s \leftarrow \poly(d/\epsilon)$
	\State Choosing a sketching matrix $S \in \R^{s \times n}$
	\State $\wt{SA} \leftarrow \textsc{SketchLog}(S,M)$
	\State $\wt{x} \leftarrow \min_{x \in \R^d } \| \wt{SA} x - S b \|_2$
	\State \Return $\wt{x}$
\EndProcedure
\end{algorithmic}
\end{algorithm}
Recall that for a matrix $M \in \R^{n\times d}$, we use $\log (M)$ to denote the $n \times d$ matrix where the entry at $i$-th row and $j$-th column of matrix $\log M$ is $\log (M_{i,j})$.

\begin{theorem}[Linear regression]\label{thm:linear_regression}
Given matrix $M \in \R^{n \times d}$ and vector $b \in \R^d$ where $n \gg d$. Let $A = \log M \in \R^{n\times d}$. 
There is an one-pass algorithm (Algorithm~\ref{alg:linear_regression}) that uses $\poly(d, \log n, 1/\epsilon)$ space, receives the update of $M$ in the stream, and outputs vector $\wt{x} \in \R^d$ such that 
\begin{align*}
\| A \wt{x} - b \|_2 \leq ( 1 + \epsilon ) \min_{x \in \R^d } \| A x - b \|_2 + \tau,
\end{align*}
holds with probability at least $9/10$ and where $\tau = \| b \|_2 / \poly(d/\epsilon)$.
\end{theorem}

\begin{proof}
Without loss of generality, we assume that $\| A \|_2 = 1$ in the proof.

Let $x^* \in \R^d$ denote the optimal solution of this problem,
\begin{align*}
\min_{x\in \R^d} \| A x - b \|_2.
\end{align*}
Let $\OPT$ denote $\| A x^* - b \|_2$.
Let $x' \in \R^d$ denote the optimal solution of this problem,
\begin{align*}
\min_{x\in \R^d} \| S A x -  S b \|_2.
\end{align*}
By property of sketching matrix, we have
\begin{align*}
 \| A x' - b \|_2 \leq (1+\epsilon) \| A x - b \|_2.
\end{align*}
Note that $x' = (SA)^\dagger Sb$. Let $\wt{x} \in \R^d$ denote the optimal solution of
\begin{align*}
\min_{x \in \R^d} \| \wt{SA} x - S b \|_2.
\end{align*}
It means $\wt{x} = ( \wt{SA} )^\dagger Sb $. We have
\begin{align}\label{eq:bound_A_wtx_b}
 & ~ \| A \wt{x} - b \|_2 \notag \\
\leq & ~ \| A x' - b \|_2 + \| A \wt{x} - A x' \|_2 \notag \\
\leq & ~ (1+\epsilon)\OPT + \| A \wt{x} - A x' \|_2 \notag \\
= & ~  (1+\epsilon)\OPT + \underbrace{ \| A (\wt{SA})^\dagger Sb - A (SA)^\dagger Sb \|_2 }_{C_1} 
\end{align}
where the first step follows by triangle inequality, the second step follows by $\| Ax'-b \|_2 \leq (1+\epsilon) \OPT$, the third step follows by definition of $x'$ and $x^*$.

Now the question is how to bound the term $C_1$ in Eq.~\eqref{eq:bound_A_wtx_b}. We can upper bound $C$ in the following way,
\begin{align*} 
C_2 = & ~ \| A (\wt{SA})^\dagger Sb - A (SA)^\dagger Sb \|_2 \\
 \leq & ~ \| A \|_2 \| (\wt{SA})^\dagger - (SA)^\dagger \|_2 \| S b\|_2 \\
 = & ~  \| (\wt{SA})^\dagger - (SA)^\dagger \|_2  \| S b \|_2  \\
 \lesssim & ~  \underbrace{ \| (\wt{SA})^\dagger - (SA)^\dagger \|_2 }_{C_2} \| b \|_2 
\end{align*}
the third step follows by $\|A \|_2=1$, and the the last step follows by $\| S b \|_2 = O(1) \cdot \| b \|_2$.
Next, we show how to bound the term $C_2$ in the above equation, using Lemma~\ref{lem:w73_lemma}, \ref{lem:perturbation_lemma} and \ref{lem:latala_lemma}, we have
\begin{align*}
C_2 = & ~ \| (\wt{SA})^\dagger - (SA)^\dagger \|_2  \\
\lesssim & ~ \max ( \| (\wt{SA})^\dagger \|_2^2, \| (SA)^\dagger \|_2^2 ) \cdot \| \wt{SA} - SA \|_2 \\
\lesssim & ~ \| (SA)^\dagger \|_2^2 \cdot \| \wt{SA} - SA \|_2 \\
\lesssim & ~ \| A^\dagger \|_2^2 \cdot \| \wt{SA} - SA \|_2 \\
\leq & ~ \| A^\dagger \|_2^2 /  \kappa^2 \poly(d/\epsilon) \\
= & \| A \|_2^2 / \poly(d/\epsilon) \\
\lesssim & ~ 1/ \poly(d/\epsilon)
\end{align*}
where the second step follows by Lemma~\ref{lem:w73_lemma}, the third step follows by Lemma~\ref{lem:perturbation_lemma}, the fourth step follows by property of sketching matrix $S$, the fifth step follows by size of $S$ and Lemma~\ref{lem:latala_lemma} the last step follows by $\| A \|_2=1$.
\end{proof}

\section{Tools}
In this section, we introduce several basic perturbation results.

\cite{w73} presented a perturbation bound of Moore-Penrose inverse the spectral norm,
\begin{lemma}[\cite{w73}, Theorem 1.1 in \cite{mz10}]\label{lem:w73_lemma}
Given two matrices $A,B \in \R^{d_1 \times d_2}$ with full column rank, we have
\begin{align*}
\| A^\dagger - B^\dagger \|_2 \lesssim \max ( \| A^\dagger \|_2^2, \| B^\dagger \|_2^2 ) \cdot \| A - B \|_2.
\end{align*}
\end{lemma}

\begin{lemma}[Latala's theorem \cite{l05}, Theorem 5.37 in \cite{v10}]\label{lem:latala_lemma}
Let $A$ be a random $n\times d$ matrix whose entries $A_{i,j}, \forall (i,j) \in [n] \times [d]$ are independent centered random variables with finite fourth moment. Then
\begin{align*}
\E[ \| A \|_2 ] \lesssim \max_{i\in [n] } \left( \sum_{j=1}^d \E [ A_{i,j}^2 ] \right)^{1/2}  + \max_{j\in [d]} \left( \sum_{i=1}^n \E [ A_{i,j}^2 ] \right)^{1/2} + \left( \sum_{i=1}^n\sum_{j=1}^d \E [ A_{i,j}^4 ] \right)^{1/4} .
\end{align*}
\end{lemma}

\begin{lemma}\label{lem:perturbation_lemma}
Let $B = A + E$, if $|E_{i,j}| \leq \epsilon'$, where $\epsilon' = \epsilon / (d_1 d_2 \kappa(A) 10 )$, then 
\begin{align*}
(1-\epsilon) \| A^\dagger \| \leq \| B^\dagger \| \leq (1+\epsilon) \| A^\dagger \|.
\end{align*}
\end{lemma}
\begin{proof}
Given the definition of $B$, we can rewrite $B B^\top$ into four terms,
\begin{align*}
BB^\top = (A+E) (A+E)^\top = AA^\top + E A^\top + A E^\top + E E^\top .
\end{align*}
Since we can bound
\begin{align*}
 \| E A^\top + A E^\top + E E^\top \|_2
\leq & ~ 3 \| E \|_F \| A \|_2 \\
\leq & ~ 3 \epsilon' d_1 d_2 \sigma_1(A) \\
\leq & ~ \epsilon \sigma_{\min}(A) /3
\end{align*}
Thus,
\begin{align*}
AA^\top - \epsilon \sigma_{\min}(A)/3\cdot I  \preceq BB^\top \preceq AA^\top + \epsilon \sigma_{\min}(A)/3 \cdot I 
\end{align*}
Thus,
\begin{align*}
\| B^\dagger \| = & ~ \frac{1}{\sigma_{\min}(B)} \\
\leq & ~ \frac{1}{\sigma_{\min}(A) - \sigma_{\min}(A) \epsilon /3 } \\
\leq & ~ (1+\epsilon) \frac{1}{\sigma_{\min}(A)} \\
= & ~ (1+\epsilon) \| A^\dagger \|_2.
\end{align*}
This completes the proof.
\end{proof}

\begin{theorem}[Generalized rank-constrained matrix approximations, Theorem 2 in \cite{ft07}]\label{thm:theorem_2_ft07}
Given matrices $A\in \R^{n\times d}$, $B\in \R^{n\times p}$, and $C\in \R^{q\times d}$, let the SVD of $B$ be $B=U_B \Sigma_B V_B^\top$ and the SVD of $C$ be $C= U_C \Sigma_C V_C^\top$. Then, 
\begin{align*}
B^\dagger ( U_B U_B^\top A V_C V_C^\top )_k C^\dagger = \underset{\rank-k~ X\in \R^{p\times q}}{\arg\min} \| A - B X C \|_F
\end{align*}
where $(U_B U_B^\top A V_C V_C^\top)_k \in \R^{p \times q}$ is of rank at most $k$ and denotes the best rank-$k$ approximation to $U_B U_B^\top A V_C V_C^\top \in \R^{n \times d}$ in Frobenius norm.
\end{theorem}



\section{Complete Experimental Results} \label{sec:complete_exp}

To demonstrate the advantage of our proposed method, we complement the theoretical analysis with empirical study on synthetic and real data. We consider the low rank approximation task with $f(x)=\log(x)$ and $f(x)=\sqrt{x}$, vary the amount of space used by our method, and compare the errors of the solutions obtained to the optimum. The we provide additional experiments testing some other aspects of the method such as robustness to the parameter values.

\paragraph{Setup}
Given a data stream in the form of $(i_t,j_t,\delta_t)$, we use the algorithm in Section~\ref{sec:lowrank} to compute the top $k=10$ singular vectors $L$, and then compare the error of this solution to the error of the optimal solution (i.e., the true top $k$ singular vectors). 
Let $A$ denote the accumulated matrix, $M=f(A)$ denote the transformed one, and $U$ denote the top $k$ singular vectors of $M$. 
Then the evaluation criterion is 
\[
  \text{error-ratio}(L) = \frac{\|M - LL^\top M\|_F}{\|M-UU^\top M\|_F}.
\]
Clearly, the error ratio is at least $1$, and a value close to $1$ demonstrates that our solution is nearly optimal.  

Besides demonstrating the effectiveness, we also exam the tradeoff between the solution quality and the space used. Recall that there is a parameter in the the sketching methods controlling the amount of space used (line 20 in \textsc{LogSum} and line 6 in \textsc{PolySum}). We vary its value, and set the parameters in other steps of our algorithm so that the amount of space used is dominated by that of the sketching. We then plot how the error ratios change with the amount of space used.
The plotted results are the average of 5 runs; the variances are too small to plot. 

Finally, we also report the results of a baseline method: uniformly at random sample a subset $T$ of columns from $A$, and then compute the top $k$ singular vectors of $f(T)$. The space occupied by the columns sampled is similar to the space required by our algorithm for fair comparison. Since our algorithm is randomized, the expected amount of space occupied is used to determine the sample size of the baseline, and is also used for the plots. 
In the experiments, the actual amount occupied is within about $10\%$ of the expected value.

\paragraph{Implementation and Parameter Setting.}
In our algorithm for low rank approximation, an FJLT matrix $S$ is used~\cite{a03,ac06}. 
In the step of adaptive sampling, instead of setting the threshold $\eta$, for simplicity we let $q_i = \max\{\tilde s_i, 0\}$ and set $p_i = q_i + \sum_i q_i / n$. 

For the sketching subroutine, instead of specifying the desired $\epsilon$, we directly set the size of the data structure, so as to exam the tradeoff between space and accuracy. 
Then we set $s = d_1 = d_2$ and set their value so that the space used in the corresponding step is at most that used by the sketch method. In particular, we set them equal to the size upper bounds in line 20 in \textsc{LogSum} or line 6 in \textsc{PolySum}.

\subsection{Synthetic Data}

\paragraph{Data Generation.}

The following data sets are generated. Note that although we don't provide theoretical analysis for $f(x) = \sqrt{x}$, one could follow that for $f(x) = \log(x)$ to get similar guarantees, and we also generate synthetic data to test our method in this case. 
\begin{enumerate}
	\item \LOGDATA: This is for the experiments with $f(x)=\log(x)$. First generate a matrix $M$ of $n\times n$ where the entries are i.i.d.\ Gaussians. To break the symmetry of the columns, scale the length of the $i$-th column to $4/i$. Finally, generate matrix $A$ with $A_{ij}=\exp(M_{ij})$. Each entry $A_{ij}$ is divided into equally into $5$ updates $(i,j,A_{ij}/5)$, and all the updates arrive in a random order. The size $n$ can be $10000$, $30000$, and $50000$. 
	\item \SQDATA: This is for the experiments with $f(x) = \sqrt{x}$. The data and update stream are generated similarly as \LOGDATA, except that $A_{ij}=M^2_{ij}$.
	We tested on sizes $n=10000$ and $n=30000$.
\end{enumerate}

\paragraph{Results.}
Figure~\ref{fig:app_syn_log} shows the results on the synthetic data \LOGDATA, and Figure~\ref{fig:app_syn_sqrt} shows those on \SQDATA. 
In general, the error ratio of our method is much better than that of the uniform sampling baseline: ours is close to 1 while that of uniform sampling is about 4. 
It also shows that our method can greatly reduces the amount of space needed by orders while merely comprising the solution quality, and this advantage is more significant on larger data sets.  For example, when $n=50000$, using space about $5\%$ of the matrix size leads to only about $5\%$ extra error over the optimum. Finally, we note that these observations are consistent on both $f(x)=\log(x)$ and $f(x)=\sqrt{x}$.

\begin{figure*}[!t]
\centering
\subfloat[\LOGDATA, $n=10^4$]{\includegraphics[height=0.27\linewidth]{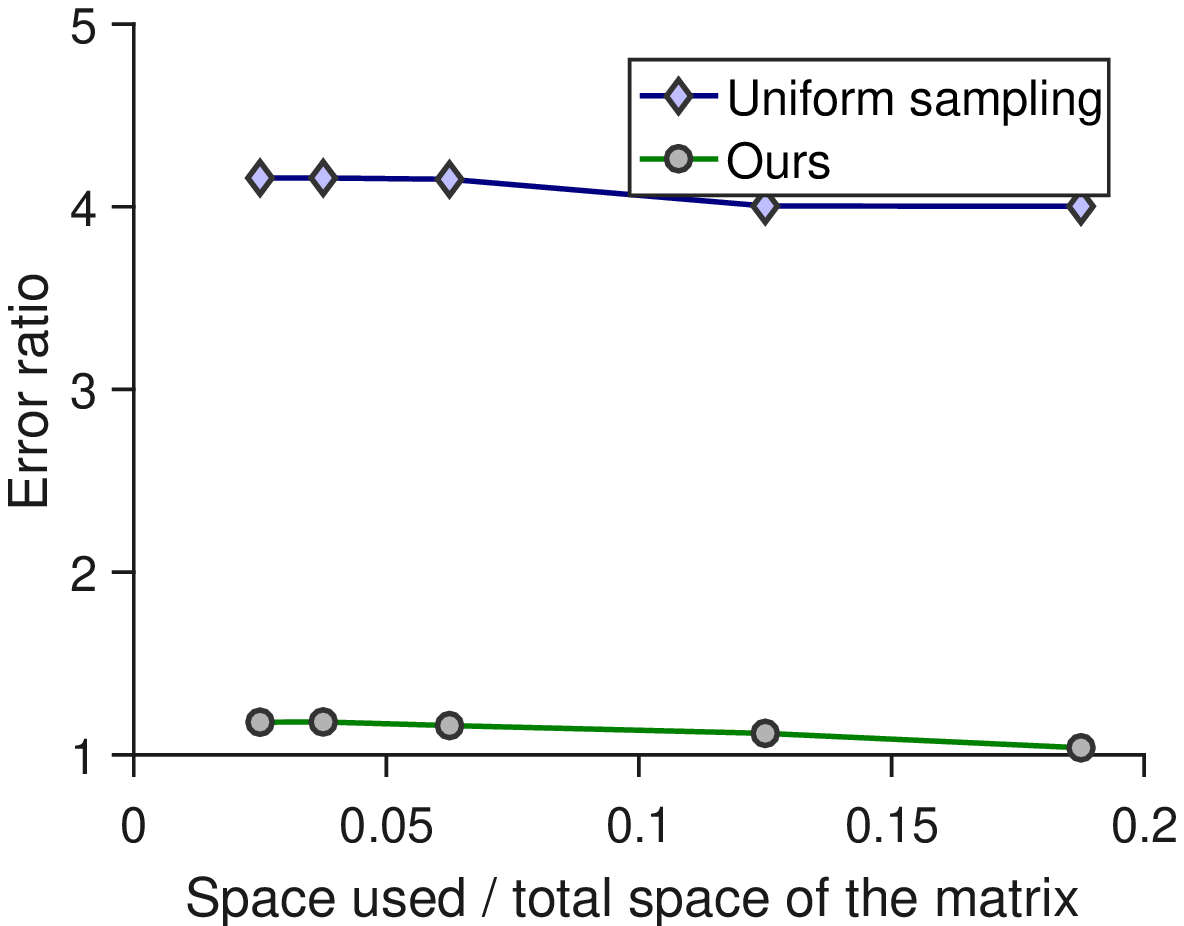}}
\subfloat[\LOGDATA, $n=3\cdot 10^4$]{\includegraphics[height=0.27\linewidth]{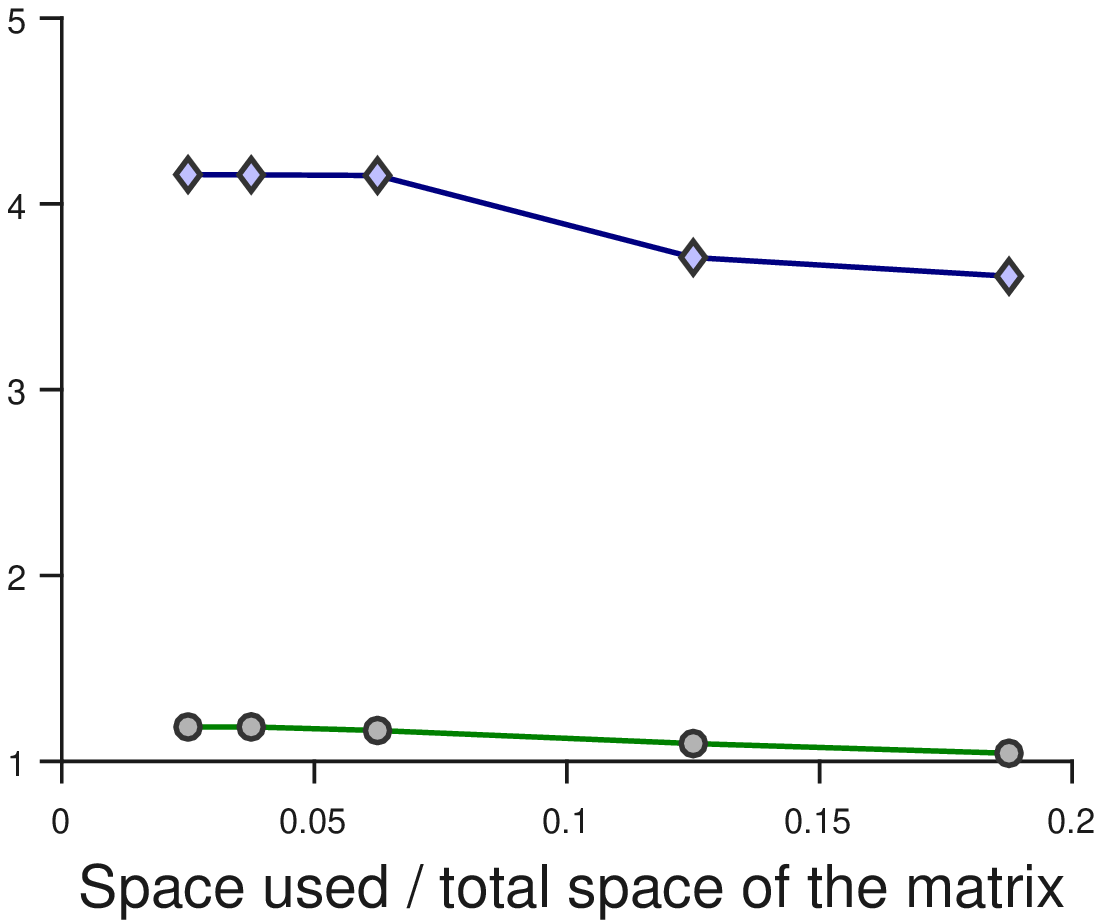}}
\subfloat[\LOGDATA, $n=5\cdot 10^4$]{\includegraphics[height=0.27\linewidth]{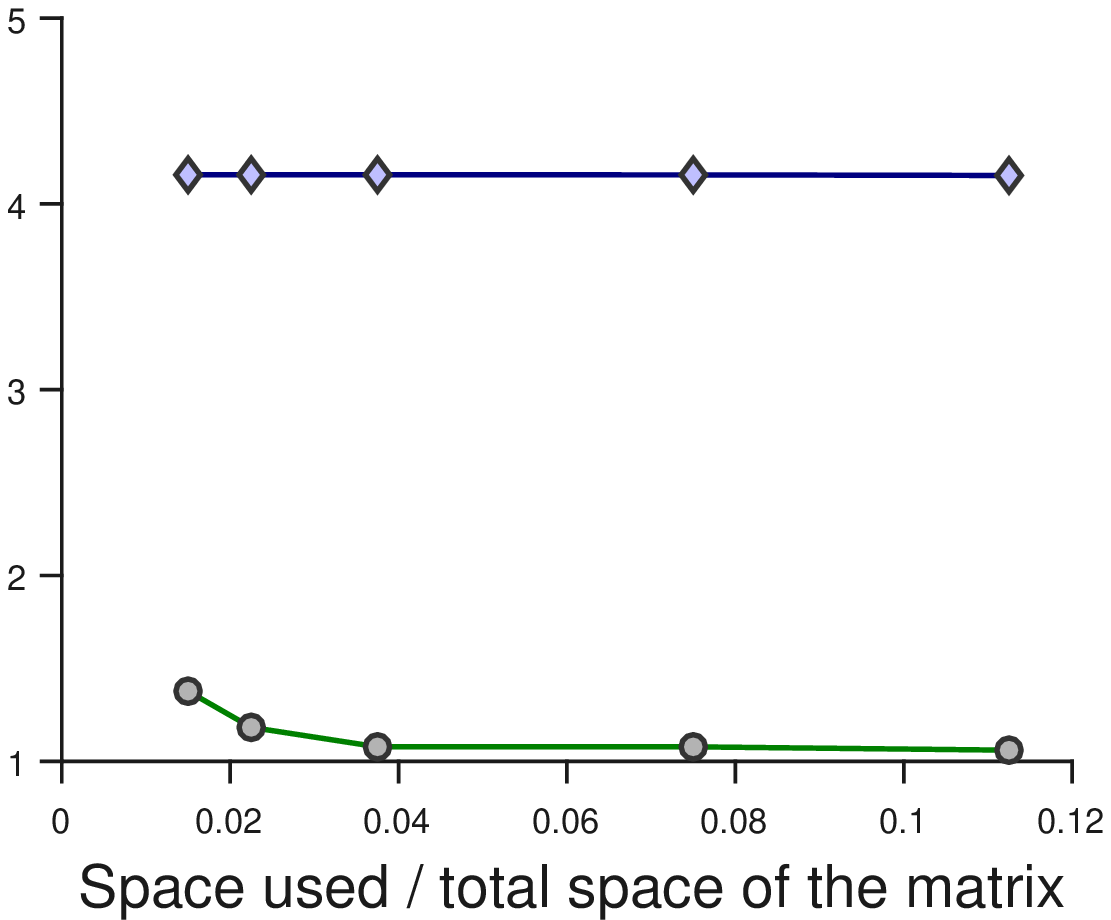}}
\caption{Error ratios on the synthetic data \LOGDATA. The $x$-axis is the ratio between the amount of space used by the algorithms and the total amount of space occupied by the data matrix. The $y$-axis is the ratio between the error of the solutions output by the algorithms and the optimal error.}
\label{fig:app_syn_log}
\end{figure*}

\begin{figure*}[!t]
\centering
\subfloat[\SQDATA, $n=10^4$]{\includegraphics[width=0.32\linewidth]{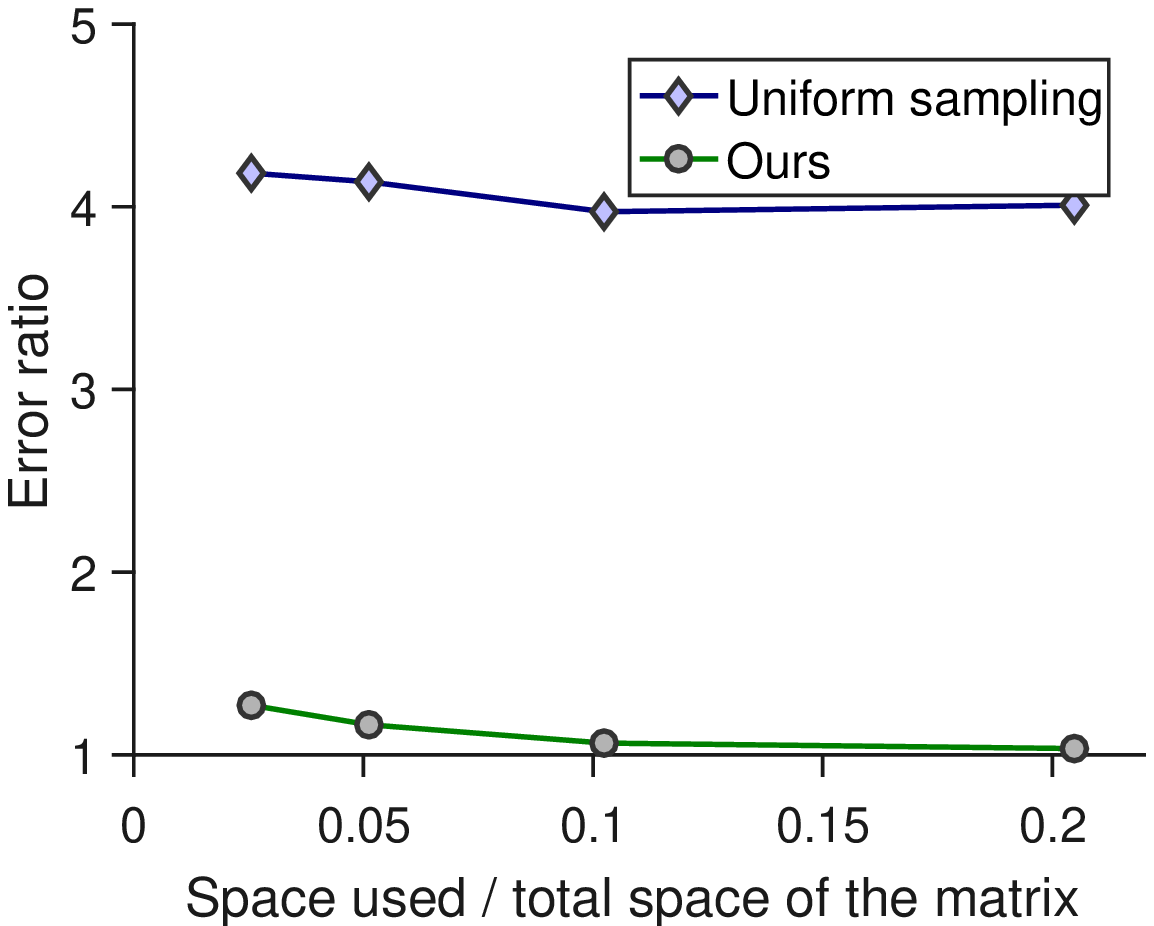}}
\subfloat[\SQDATA, $n=3\cdot 10^4$]{\includegraphics[width=0.32\linewidth]{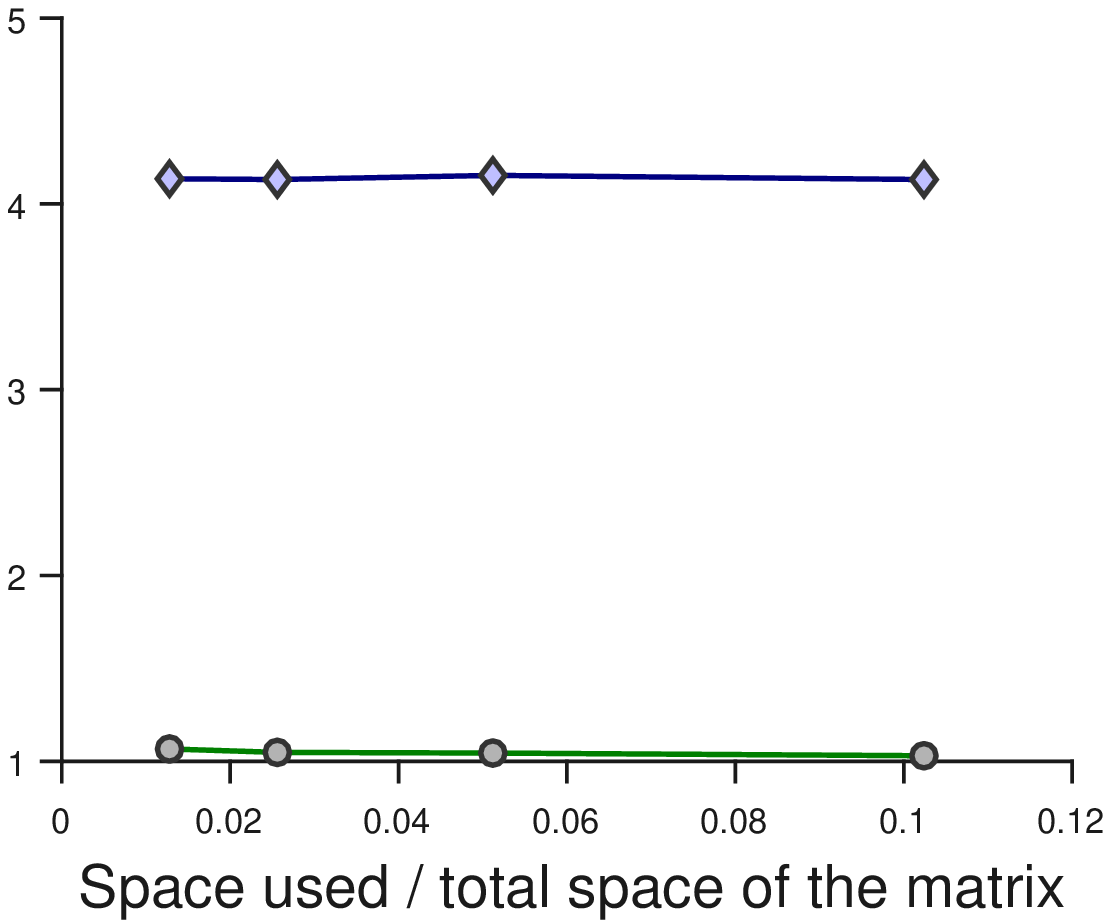}}
\subfloat[\SQDATA, $n=5\cdot 10^4$]{\includegraphics[width=0.32\linewidth]{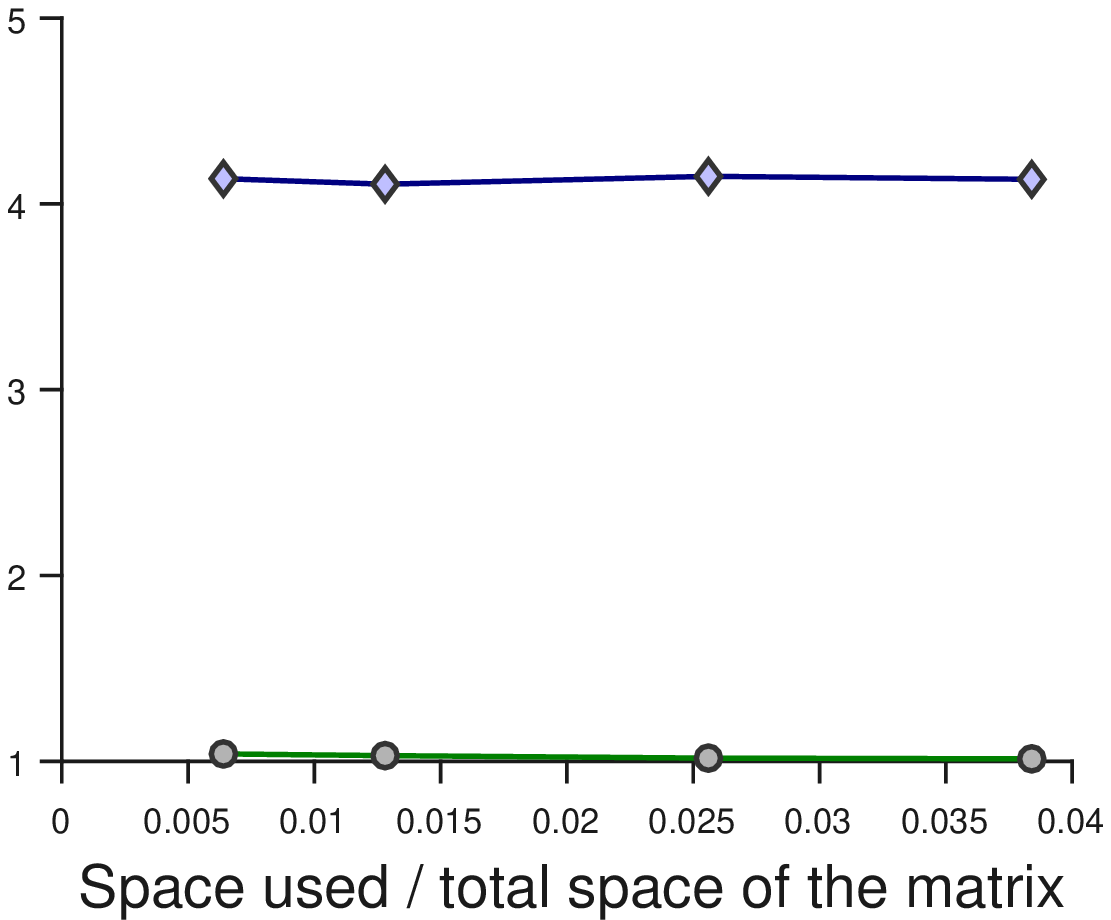}}
\caption{Error ratios on the synthetic data \SQDATA. The $x$-axis is the ratio between the amount of space used by the algorithms and the total amount of space occupied by the data matrix. The $y$-axis is the ratio between the error of the solutions output by the algorithms and the optimal error. }
\label{fig:app_syn_sqrt}
\end{figure*}

\begin{figure*}[!t]
\centering
\subfloat[Real data, $n=10^4$]{\includegraphics[height=0.27\linewidth]{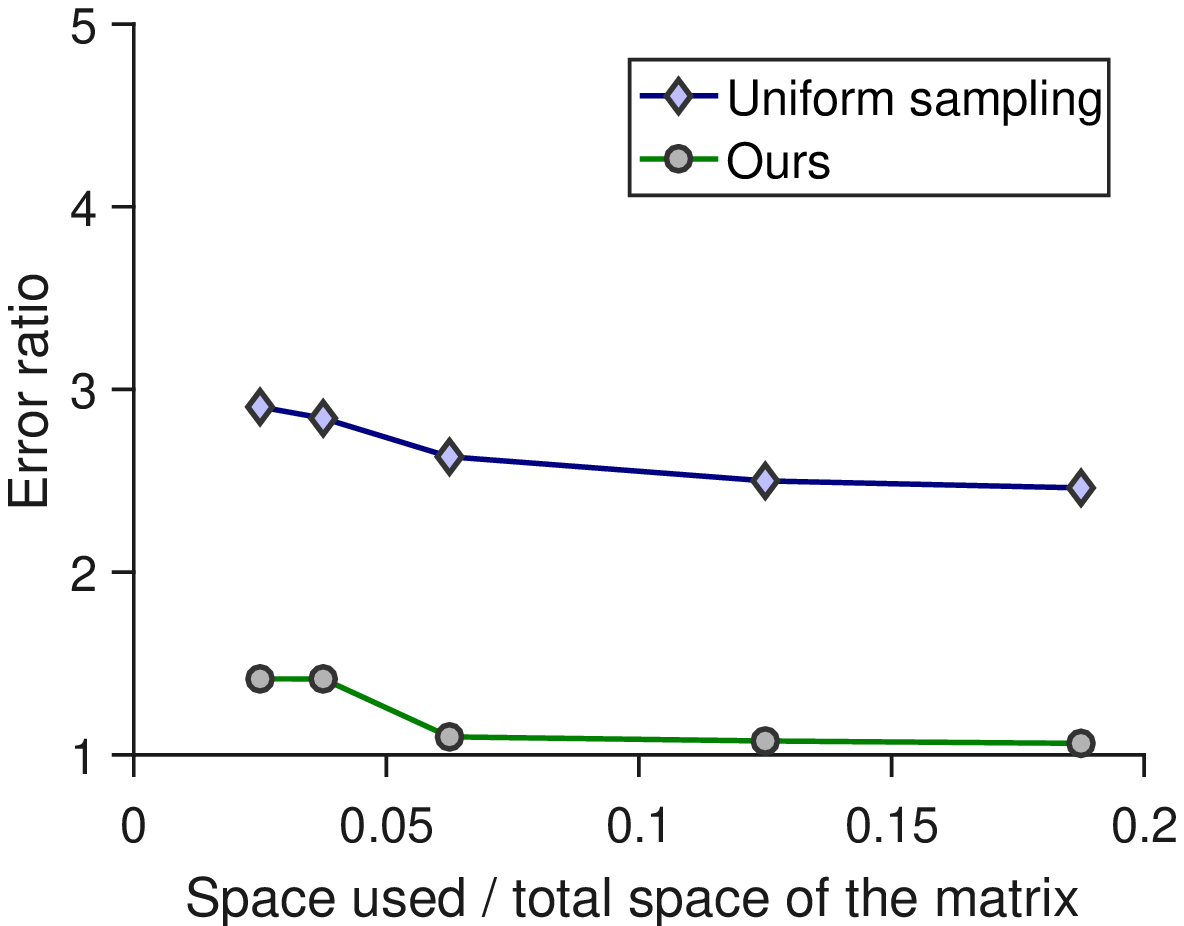}}
\subfloat[Real data, $n=3\cdot10^4$]{\includegraphics[height=0.27\linewidth]{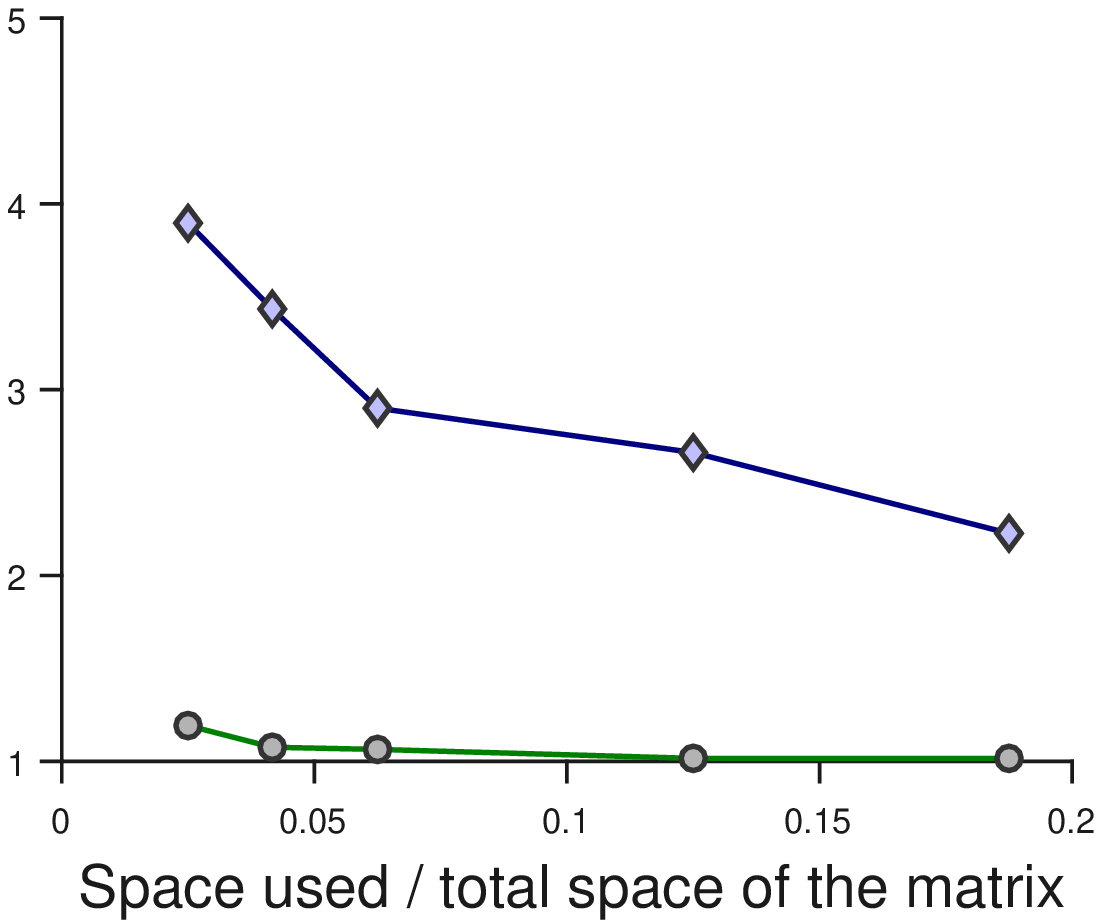}}
\subfloat[Real data, $n=5\cdot10^4$]{\includegraphics[height=0.27\linewidth]{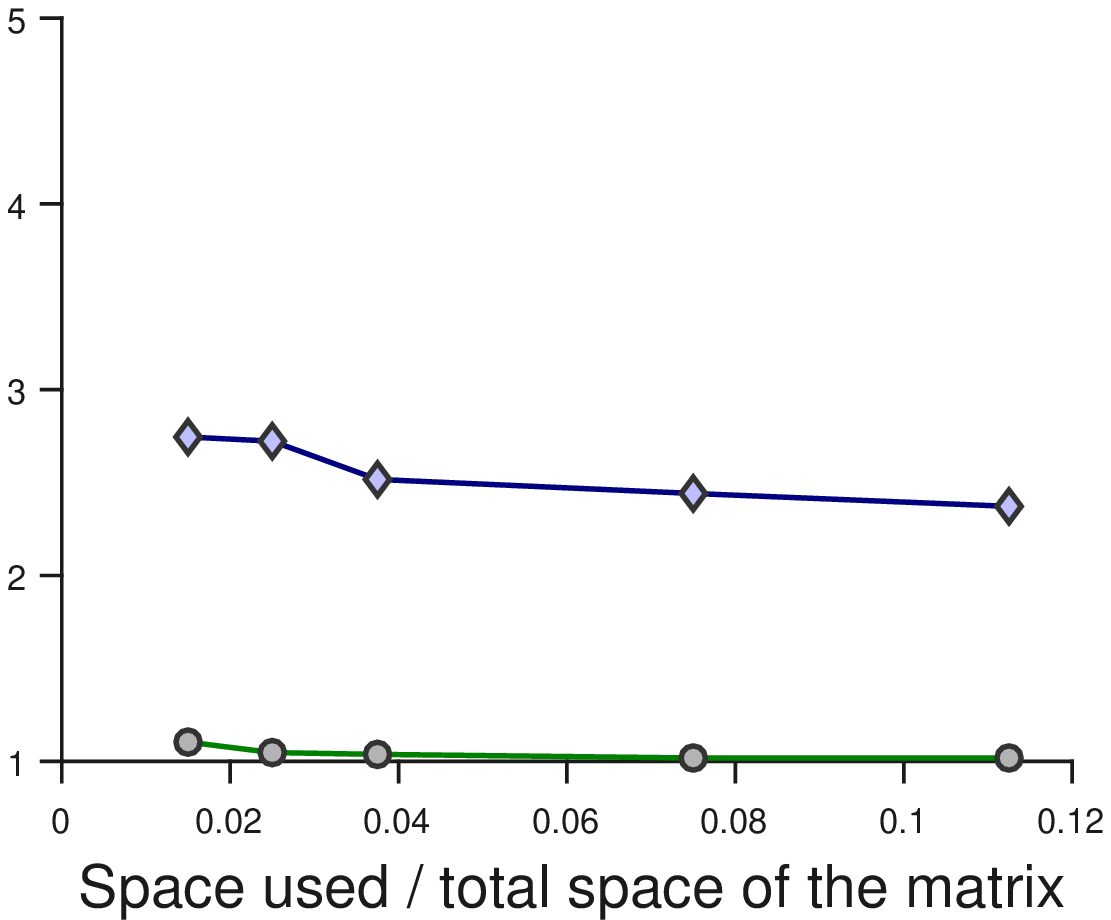}}
\caption{Error ratios on the real data (Wikipedia). The $x$-axis is the ratio between the amount of space used by the algorithms and the total amount of space occupied by the data matrix. The $y$-axis is the ratio between the error of the solutions output by the algorithms and the optimal error. } 
\label{fig:app_real}
\end{figure*}


\subsection{Real Data}
We exam our method on the real world data from the NLP application word embedding, which is a motivating example for proposing our approach.
Our method with $f(x)=\log(x+1)$ is used. The parameters are set in a similar way as for the synthetic data.

\paragraph{Data Collection.}
The data set is the entire Wikipedia corpus~\cite{enwiki}
consisting of about 3 billion tokens. Details can be found in the appendix and only a brief description is provided here. The matrix to be factorized is $M$ 
with $M_{ij} = p_j\log \frac{N_{ij} N }{N_i N_j}$ where 
$N_{ij}$ is the number of times words $i$ and $j$ co-occur in a window of size $10$, $N_i$ is the number of times word $i$ appears, $N$ is the total number of words in the corpus, and $p_j$ is a weighting factor depending set to $p_j = \max\{1, (N_j/N_{10})^2\}$, which puts larger weights on more frequent words since they are less noisy~\cite{pennington2014glove,lg14}. 
Note that $N_i$'s and $N$ can be computed easily, so essentially the only dynamically update part is $\log N_{ij}$.

The data stream is generated as a window of size 10 slides along the sentences in the corpus and we collect the co-occurrence counts of the word pairs in the window. Here the count is weighted, i.e., if two words appear in a distance of $t$ inside the window, then the count update value is $1/t$ as in~\cite{pennington2014glove}.  
We consider the matrix for the most frequent $n$ words, where $n=10000$, $30000$, and $50000$.

\paragraph{Results.}
Figure~\ref{fig:app_real} shows the results on the real data. The observations are similar to those on the synthetic data: the errors of our method are much better than the baseline, and are close to the optimum; the method is very space efficient without increasing the error much. These results again demonstrate its effectiveness.
\subsection{The Effect of the Sample Size}

In our algorithm we have parameters $s = d_1 = d_2$ that determine the sample sizes in different steps of the algorithm. In previous experiments, we set them equal to the size upper bounds in line 20 in \textsc{LogSum} or line 6 in \textsc{PolySum}. 
Here we consider varying their values. In particular, we use the Wikipedia data with $n=10000$ and $f(x)=\log(x)$ and set the size upper bound of the sketch method to be $200$. Then we set $s = d_1 = d_2 = \gamma$ and vary the value of $\gamma$. 

\paragraph{Results.}
Figure~\ref{fig:app_nsample} shows the results with various sample sizes. It is observed that smaller sample sizes lead to worse errors as expected, but overall the results are quite stable across different sizes. This demonstrates the robustness of our method to these parameters. It is also observed that after a certain value, increasing the sample size doesn't lead to better error, which should be due to the approximation error introduced by the sketch. The results suggest that in general the sample size should be set approximately equal to the size upper bound in the sketch method. 

\begin{figure*}[!t]
\centering
\includegraphics[width=0.37\linewidth]{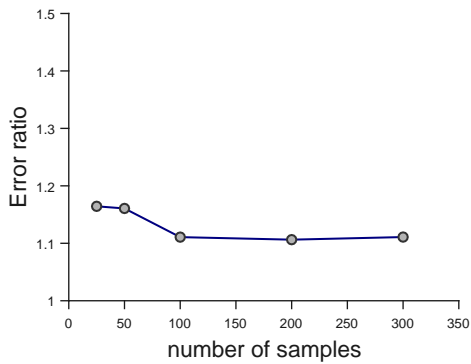}
\caption{Error ratios when using different sample size in the algorithm. The $x$-axis is the ratio between the amount of space used by the algorithms and the total amount of space occupied by the data matrix. The $y$-axis is the ratio between the error of the solutions output by the algorithms and the optimal error. } 
\label{fig:app_nsample}
\end{figure*}

\end{document}